\documentclass{amsart}

\usepackage[linesnumbered,ruled,vlined]{algorithm2e}
\usepackage{algpseudocode}

\SetKwInput{KwInput}{Input}
\SetKwInput{KwOutput}{Output}

\usepackage[american]{babel}

\usepackage[normalem]{ulem} 

\usepackage[letterpaper,top=2cm,bottom=2cm,left=3cm,right=3cm,marginparwidth=1.75cm]{geometry}

\usepackage{comment}

\usepackage{enumerate}

\usepackage{amsmath}
\usepackage{amsthm}
\usepackage{graphicx}
\usepackage[colorlinks=true, allcolors=red]{hyperref}
\usepackage[numbers,sort]{natbib}

\usepackage[american]{babel}	
\usepackage{csquotes}	
\usepackage{graphicx}			
\usepackage{epstopdf}
\usepackage{subfig}				
\usepackage{marvosym}		
\usepackage{amssymb}
\usepackage{mathtools}
\usepackage{xcolor}

\usepackage{tikz,fullpage}
\usetikzlibrary{arrows,%
        petri,%
        topaths,%
        positioning,%
        shapes.geometric}
\usepackage{tkz-berge}

\newtheorem{theorem}{Theorem}
\newtheorem{lemma}[theorem]{Lemma}
\newtheorem{proposition}[theorem]{Proposition}
\newtheorem{corollary}[theorem]{Corollary}

\theoremstyle{definition}
\newtheorem{definition}[theorem]{Definition}

\newtheorem{example}[theorem]{Example}

\title[Enumerating monophyletic characters]{Enumerating monophyletic characters in mathematical phylogenetics}

\author[Fischer]{Mareike Fischer}
\address{Institute of Mathematics and Computer Science, Greifswald University, Greifswald, Germany} \email{mareike.fischer@uni-greifswald.de, email@mareikefischer.de}

\author[Hendriksen]{Michael Hendriksen}
\address{School of Mathematics and Statistics, University of New South Wales, Sydney, NSW, Australia}
\email{m.hendriksen@unsw.edu.au}

\author[Wicke]{Kristina Wicke}
\address{Department of Mathematical Sciences, New Jersey Institute of Technology, Newark, NJ, USA and National Institute for Theory and Mathematics in Biology, Northwestern University and The University of Chicago, Chicago, IL, USA}
\email{kristina.wicke@njit.edu}

\definecolor{changecolor}{RGB}{192,64,0}

\usepackage[normalem]{ulem}
 
\begin{document}

\begin{abstract} 
Grouping species according to their phylogenetic relationships often results in different groups than grouping them according to their shared traits. Monophyletic groups play an important role in this regard, as they are groups of species sharing the same trait and being uniquely defined by a joint phylogenetic subtree. This immediately leads to the question of how to identify possible monophyletic groups in characters, which assign each present-day species a certain trait and which are typically used for phylogenetic tree reconstruction.

In our manuscript, we provide a general formula to quantify how many different characters are monophyletic on any given tree and provide simple formulae for binary characters and for certain tree shapes. We also investigate relations between monophyly and the well-known phylogenetic tree reconstruction criterion maximum parsimony by providing a linear-time algorithm which determines the parsimony score together with the monophyly type of a character on a tree.
 \smallskip \newline
\noindent \textbf{Keywords.} 
monophyly, phylogenetic tree, characters, maximum parsimony
\smallskip \newline
\noindent \textbf{MSC identifiers.} 05A15, 05C05, 92B05
\end{abstract}

\maketitle

\section{Introduction}
Investigating the relationships between different species is one of the main aims in evolutionary biology and mathematical phylogenetics. One way this investigation can be done is by analyzing \emph{phylogenetic trees}, which are often used to depict evolutionary relationships and which can be reconstructed based on aligned sequence data such as DNA, RNA, or protein sequences~\cite{felsenstein2003inferring,Semple2003}. Another way to understand the relationships between species, however, is the investigation of certain \emph{traits} -- i.e., of shared characteristics among different species. Interestingly, these two viewpoints do not always lead to the same outcome: for instance, it is possible that a trait is shared between two species which are not closely related according to the underlying phylogenetic tree (this happens, for instance, in the case of convergent evolution \cite{speed}), whereas two closely related species might not share the same trait (which happens, for instance, in case of gene loss \cite{Losos2011}). Such discrepancies arise because the evolution of traits may involve convergence, reversals, or uncertainty in ancestral states. As an example, although morphologically defined mammalian and molluscan genera have been shown to be significantly more likely to be monophyletic with respect to molecular phylogenies than expected by chance \cite{Jablonski2009}, the correspondence between trait-based classifications and phylogenetic relationships is not perfect, motivating a systematic study of when observed molecular data admit monophyletic interpretations. It is the main aim of this manuscript to investigate and enumerate settings in which aligned sequence data leads to monophyletic groups of species.

Roughly speaking, a group of species is called \emph{monophyletic} concerning a certain trait if it has a common ancestor and if it contains precisely those species (from the present and the past) sharing this trait. In a phylogenetic tree, a monophyletic group can be recognized by the presence of a vertex all of whose descendants share the trait under investigation, whereas this trait cannot be found anywhere else in the tree. As an example, considering tree $T$ and the traits assigned to its leaves in the row labeled by $f_1$ in Figure~\ref{fig_monotypes}, species 1 and 2 share trait $\alpha$ and form a monophyletic group. This can be seen by assigning trait $\alpha$ to vertex $u$ and assigning a different trait to all other inner vertices.

Aligned sequence data, which is often used to reconstruct trees, typically only contains present-day data. For instance, a DNA alignment usually contains DNA data of species living today. The tree reconstructed from these data contains information about common ancestry, but even in cases where this tree is known, there might be some uncertainty concerning the traits of (typically extinct) ancestors, for which usually no sequence data are available. For instance, if two sister species have different traits, certain phylogenetic reconstruction criteria like maximum parsimony (MP) \cite{Fitch,Hartigan1973} might consider both traits as equally likely for the common ancestor. 

Previous work on the concept of monophyly in phylogenetics includes studies of monophyly under neutral models of evolution (e.g., \cite{Zhu2011,Mehta2016}), and, very recently, the use of monophyletic groups in partitioning phylogenetic trees into clusters~\cite{Ganesan2025}. However, to the best of our knowledge, no previous work has systematically enumerated character patterns according to whether they admit monophyletic realizations on a fixed phylogenetic tree. Such an enumeration provides a better understanding of how restrictive the requirement of monophyly is for molecular characters on a given tree, thereby establishing a combinatorial baseline to which observed data can be compared. This combinatorial perspective is the focus of the present manuscript.

We investigate so-called \emph{characters} (also often referred to as \emph{sites}), which are basically columns in sequence alignments assigning a trait to each species under investigation. We call a character \emph{monophyletic} on a given tree $T$ if it permits a trait assignment to the internal vertices of tree $T$ such that at least one present-day species group (which may consist of only one species) is monophyletic. 
Moreover, we call it \emph{fully monophyletic} on $T$ if every present-day species belongs to a monophyletic group. 
For instance, considering again tree $T$ from Figure \ref{fig_monotypes}, character $f_1=\alpha \alpha \beta \gamma \beta \gamma$, which assigns trait $\alpha$ to species 1 and 2, trait $\beta$ to species 3 and 5, and trait $\gamma$ to species 4 and 6, is monophyletic on $T$, with species 1 and 2 forming a monophyletic group. 
However, $f_1$ is not fully monophyletic, as there is no way of choosing states for all internal vertices to make species with trait $\beta$ (namely 3 and 5) or trait $\gamma$ (namely 4 and 6) members of a monophyletic group.  Character $f_2=\alpha \alpha \beta \gamma \gamma \gamma$, on the other hand, is fully monophyletic on $T$, which can be seen when assigning trait $\alpha$ to $u$ and trait $\gamma$ to $v$ and $\beta$ to all other inner vertices, as then all groups of at least two present-day species with the same trait, i.e., the groups $\{1,2\}$ (trait $\alpha$) as well as $\{4,5,6\}$ (trait $\gamma$) are monophyletic.

Note that (fully) monophyletic characters that induce only singleton monophyletic groups provide no information about shared ancestry among multiple taxa. In particular, such characters do not impose meaningful constraints on ancestral state assignments or tree structure, as any singleton state can trivially be realized on any tree.
For these reasons, we say that a monophyletic character is \emph{relevant} if it induces at least one relevant monophyletic group, whereas a fully monophyletic character is relevant if all its monophyletic groups are relevant.
For example, the monophyletic character $f_1$ in Figure~\ref{fig_monotypes} is relevant as it induces a relevant monophyletic group, namely the group $\{1,2\}$. The fully monophyletic character $f_2$, on the other hand, is not relevant. While the monophyletic groups $\{1,2\}$ and $\{4,5,6\}$ have sizes two and three, respectively, the remaining monophyletic group, group $\{3\}$, only has size one.
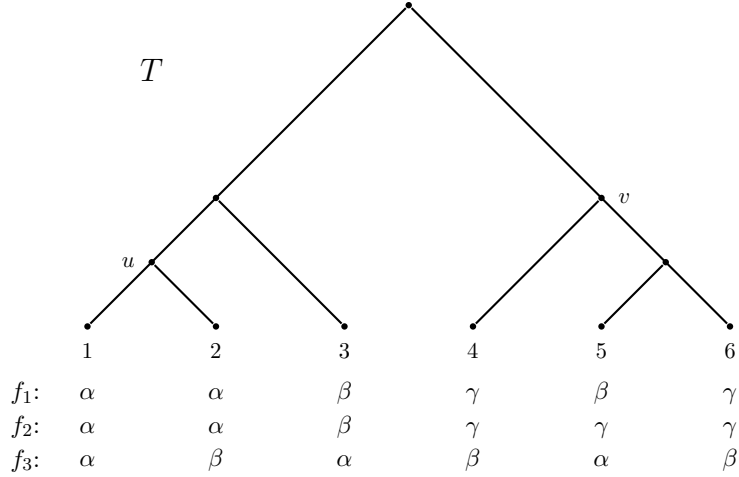
\begin{figure}
    \centering
    \begin{tikzpicture}[
    sdot/.style={circle,fill,radius=1pt,inner sep=1pt},
    thick,scale=0.85, transform shape]

\tikzset{
  char/.style={
    font=\large,
    text height=1.5ex,
    text depth=0.25ex
  }
}

\node at (1,4){\LARGE $T$};

\node[sdot] (1) at (0,0) {}; 
\node[sdot] (2) at (2,0) {}; 
\node[sdot] (3) at (4,0) {}; 
\node[sdot] (4) at (6,0) {}; 
\node[sdot] (5) at (8,0) {}; 
\node[sdot] (6) at (10,0) {}; 

\node[sdot] (A) at (1,1) {}; 
\node[sdot] (B) at (9,1) {}; 
\node[sdot] (C) at (2,2) {}; 
\node[sdot] (D) at (8,2) {}; 
\node[sdot] (E) at (5,5) {}; 

\draw (1)--(E)--(6);
\draw (A)--(2);
\draw (B)--(5);
\draw (C)--(3);
\draw (D)--(4);

\node[below=2pt of 1] {1};
\node[below=2pt of 2] {2};
\node[below=2pt of 3] {3};
\node[below=2pt of 4] {4};
\node[below=2pt of 5] {5};
\node[below=2pt of 6] {6};

\node[left=2pt of A] {$u$};
\node[right=2pt of D] {$v$};

\node[below=20pt of 1, char] (f1) {$\alpha$};
\node[below=20pt of 2, char] {$\alpha$};
\node[below=20pt of 3, char] {$\beta$};
\node[below=20pt of 4, char] {$\gamma$};
\node[below=20pt of 5, char] {$\beta$};
\node[below=20pt of 6, char] {$\gamma$};
\node[left=10pt of f1]{\large $f_1$:};

\node[below=35pt of 1, char] (f2) {$\alpha$};
\node[below=35pt of 2, char] {$\alpha$};
\node[below=35pt of 3, char] {$\beta$};
\node[below=35pt of 4, char] {$\gamma$};
\node[below=35pt of 5, char] {$\gamma$};
\node[below=35pt of 6, char] {$\gamma$};
\node[left=10pt of f2]{\large $f_2$:};

\node[below=50pt of 1, char] (f3) {$\alpha$};
\node[below=50pt of 2, char] {$\beta$};
\node[below=50pt of 3, char] {$\alpha$};
\node[below=50pt of 4, char] {$\beta$};
\node[below=50pt of 5, char] {$\alpha$};
\node[below=50pt of 6, char] {$\beta$};
\node[left=10pt of f3]{\large $f_3$:};
\end{tikzpicture}
\caption{Different types of monophyly: Character $f_1=\alpha \alpha \beta \gamma \beta \gamma$ is monophyletic on $T$ but not fully monophyletic. It is relevant as it induces a relevant monophyletic group, namely $\{1,2\}$. Similarly, character $f_2=\alpha \alpha \beta \gamma \gamma \gamma$ is relevant monophyletic. Character $f_2$ is also fully monophyletic but not relevant fully monophyletic (as the monophyletic group of species with trait $\beta$, namely $\{3\}$, has size one). Finally, character $f_3=\alpha \beta \alpha \beta \alpha \beta$ has none of these properties.}
\label{fig_monotypes}
\end{figure}

Using these definitions, we will quantify how many characters have the property of being (relevant) monophyletic and (relevant) fully monophyletic, respectively. Interestingly, these counts do not only depend on the number of species under investigation, but also on some structural properties of the underlying phylogenetic tree. Moreover, the number of such characters turns out to strongly depend on how many different traits are present in the tree. We additionally show some interesting links between the different types of monophyletic characters and the famous maximum parsimony tree reconstruction criterion as well as Fitch's algorithm for ancestral state reconstruction \cite{Fitch,Hartigan1973}.

The remainder of this manuscript is organized as follows. In Section \ref{s:defbas}, we introduce definitions and basic concepts for phylogenetic trees, as well as define four different types of monophyly. In Section \ref{s:genform}, we introduce a general formula for counting two types of monophyly, and derive some general results. In Section \ref{s:treeshape}, we explicitly calculate the four monophyly numbers for the special cases of caterpillar and fully balanced trees. In Section \ref{s:bincha}, we fully characterise the four monophyly numbers for binary characters. In Section \ref{s:expmon}, we calculate the expected monophyly number for binary characters for two of the types of monophyly under two different phylogenetic models -- the uniform model \cite{Semple2003} and the Yule model \cite{yule}. In Section \ref{s:pars}, we link monophyletic characters and the famous maximum parsimony tree reconstruction criterion as well as Fitch's algorithm for ancestral state reconstruction by providing an algorithm determining the parsimony score together with the monophyly type of a character in linear time. Finally, we end with a discussion and future research possibilities for monophyletic numbers.

\section{Definitions and basic concepts}\label{s:defbas}

Before we can state our results, we first introduce some basic definitions and notions, following standard phylogenetic notation as in \cite{Semple2003, Steel2016,fischer_tree_2023}.

\subsection*{Phylogenetic trees and tree shapes} First, we need to describe the structures and objects we are interested in. We start with trees. A \emph{phylogenetic $X$-tree} is a connected and acyclic graph $T=(V,E)$ whose leaf set $V^1(T)$ (i.e., vertices of degree at most 1) is bijectively labeled by some label set $X$. All non-leaf vertices of $T$ are summarized in the set $\mathring{V}(T)$ of inner vertices. Throughout our manuscript, we simply assume $X=\{1,\ldots,n\}$. A phylogenetic $X$-tree $T$ is called \emph{rooted} if it has a designated root vertex $\rho_T$. A rooted phylogenetic $X$-tree $T$ with $|X|=n\geq 2$ is called \emph{binary} if all inner vertices have degree 3, except the root $\rho_T$, which has degree 2. Such a tree $T$ can be decomposed into its two \emph{maximal pendant subtrees} (a procedure often coined \emph{standard decomposition}), which are the two subtrees $T_1$ and $T_2$ rooted at the children of the root of $T$. In this case, $T$ is often denoted as $T=(T_1,T_2)$. More generally, given a rooted tree $T$ and a vertex $v \in V(T)$, we denote by $T_v$ the pendant subtree of $T$ rooted at $v$. Moreover, note that a rooted binary phylogenetic $X$-tree comes with an inherent hierarchy, in which a vertex $u$ is called an \emph{ancestor} of a vertex $v$ if $u$ lies on the unique path from $\rho_T$ to $v$ in $T$. In this case, $v$ is called a \emph{descendant} of $u$. Note in particular that we permit a vertex $u$ to be its own ancestor and descendant. 
Moreover, if an ancestor $u$ of $v$ is adjacent to $v$, $u$ is said to be the \emph{parent} of $v$ and $v$ is the \emph{child} of $u$. A pair of leaves $u,v \in V^1(T)$ that share a common parent is called a \emph{cherry}, denoted by $[u,v]$.

In our manuscript, we only consider rooted binary phylogenetic $X$-trees, and whenever there is no ambiguity, we refer to them simply as \emph{trees}. Whenever, in such a tree, the leaf labels are disregarded, we will refer to the \emph{tree shape}, i.e., a tree shape is the graph theoretic tree underlying a rooted binary phylogenetic $X$-tree. 

In this regard, there are two important tree shapes that we consider in the present manuscript: for $n\geq 2$, the \emph{caterpillar tree} or simply caterpillar $T_n^{cat}$ is the unique tree shape with only one so-called cherry, and for $n=2^h$ (with $h\geq 1$), the \emph{fully balanced tree} $T_h^{fb}$ is the unique tree with $n$ leaves that has height $h$, where the \emph{height} of a rooted tree is the maximum number of edges on any path from the root of $T$ to a leaf of $T$ (all other trees with $n=2^h$ leaves have a larger height than $T_h^{fb}$).

We focus on the caterpillar and fully balanced trees because they are the extremal rooted binary trees with respect to tree balance~\cite{fischer_tree_2023}. The caterpillar is universally regarded as the least balanced rooted binary tree for any given number of leaves, whereas the fully balanced tree (when it exists) is the most balanced. 
Consequently, these two tree shapes will serve as natural representatives of the range of behaviours that can occur with respect to the number of monophyletic characters.

\subsection*{Monophyletic groups and monophyletic characters} Next, in order to study monophylies, we need to introduce the data we analyse on trees in order to identify them. The data we analyse are given in the form of \emph{characters}: A character $f:X\rightarrow C$ is a function assigning each leaf label (or, by the above mentioned bijection: each leaf) of a given tree $T$ a \emph{character state} or \emph{trait} of some set $C$. If $|C|=2$, i.e., if $f$ employs up to two traits, $f$ is called \emph{binary}.\footnote{Note that in some contexts, for instance whenever bipartitions of the leaf set induced by edges of a phylogenetic tree are considered, characters called binary usually employ precisely two states (see, for instance, \cite{fischer2016,fischer2019}). In our setting, however, a constant character employing only one state would also be considered binary.} 

We will also define a notation for characters. Recall that $X = \{1,\dots, n\}$, and fix some character $f$ for which $f(1)=c_1,f(2)=c_2,\dots, f(n)=c_n$, for some $c_1,\dots,c_n \in C$. Then we will write $f$ simply as the sequence $c_1\dots c_n$.

We now consider some fixed tree $T$ and character $f$. Denote the set of leaf descendants of a vertex $v$ by $L(v)$, where if $v$ is a leaf then $L(v) =\{v\}$. More generally, for a set of vertices $S$, let 
\[ L(S) = \bigcup_{v \in S} L(v),\]
i.e., $L(S)$ is the set of all leaf descendants of the vertices in $S$.

We are now in a position to define the main concepts for the present manuscript.

\begin{definition}
    Let $T$ be a tree on $X$, $f$ be a character on $T$, and $Y \subseteq X$. If there exists some vertex $v$ such that $L(v)=Y$, and for all $y \in Y$ we have that $f(y)=a$ for some fixed $a \in C$, then $Y$ is referred to as a \emph{monophyletic group} with respect to $T$ and $f$ (or just a monophyletic group, if $T$ and $f$ are clear from context). If, furthermore, $|Y| \ge 2$, we refer to $Y$ as a \emph{relevant} monophyletic group.
\end{definition}

We now define four types of monophyletic characters.

\begin{definition}
    Let $T$ be a tree on $X$, and $f$ a character on $T$. Then a character $f$ is called \emph{(relevant) monophyletic} if there exists some $Y \subseteq X$ that is a (relevant) monophyletic group with respect to $T$ and $f$.

    Furthermore, if every leaf $\ell \in X$ is a  member of some (relevant) monophyletic group, $f$ is called \emph{(relevant) fully monophyletic}.
\end{definition}

\begin{example}
    Fix $T=T_1$ from Figure \ref{fig:firstexample}, and let $\alpha, \beta, \gamma, \delta \in C$ be four distinct traits. Then the character $\alpha \beta \gamma \delta$ is monophyletic (take, for example, $Y=\{1\}$) and fully monophyletic (as every leaf is part of a singleton monophyletic group -- the set containing exactly the leaf itself), but neither relevant monophyletic nor relevant fully monophyletic, as there are no monophyletic groups of size at least two.

    The character $\alpha \alpha \beta \gamma$ is monophyletic and relevant monophyletic (take, for example, $Y=\{1,2\})$ and is fully monophyletic as again every leaf is part of a singleton monophyletic group. However, it is not relevant fully monophyletic, as not every leaf is part of a monophyletic group of size at least two.

    We consider two final extremal cases. The character $\alpha \beta \alpha \beta$ has none of the four types of monophyly, whereas the character $\alpha \alpha \beta \beta$ possesses all four.
\end{example}

In the following, the number of relevant monophyletic characters on a given tree $T$ using a trait set $C$ with $|C|=c$ will be referred to as the \emph{relevant monophyly number} of $T$ and $c$, and this number will be denoted by $m(T,c)$. The number of relevant fully monophyletic characters on a given tree $T$ using a trait set $C$ with $|C|=c$ will be referred to as the \emph{relevant full monophyly number} of $T$ and $c$, and this number will be denoted by $fm(T,c)$. The analogous numbers in which non-relevant monophyletic characters and non-relevant fully monophyletic characters are allowed are denoted $m_N(T,c)$ and $fm_N(T,c)$, respectively. We summarise this and further notation in Table~\ref{tab:mono-counts}.

\begin{table}[h]
\centering
\caption{Summary of monophyly numbers for a tree $T$ and trait set size $c$.}
\label{tab:mono-counts}
\begin{tabular}{|c|l|p{8cm}|}
\hline
\textbf{Notation} & \textbf{Name} & \textbf{Description} \\
\hline
$m(T,c)$ & relevant monophyly number & Number of characters $f$ on $T$ with $|C|=c$ such that $f$ is relevant monophyletic; i.e., there exists a subset $Y \subseteq X$ that is a relevant monophyletic group with respect to $T$ and $f$. \\
\hline
$fm(T,c)$ & relevant full monophyly number & Number of characters $f$ on $T$ with $|C|=c$ such that $f$ is relevant fully monophyletic; i.e., every leaf $\ell \in X$ belongs to some relevant monophyletic group. \\
\hline
$m_N(T,c)$ & monophyly number (non-relevant) & Number of characters $f$ on $T$ with $|C|=c$ such that $f$ is monophyletic (not necessarily relevant); i.e., there exists a monophyletic group $Y \subseteq X$ with respect to $T$ and $f$. \\
\hline
$fm_N(T,c)$ & full monophyly number (non-relevant) & Number of characters $f$ on $T$ with $|C|=c$ such that $f$ is fully monophyletic (not necessarily relevant); i.e., every leaf belongs to some monophyletic group. \\
\hline
$m_{N,v}(T,c)$ & monophyly count of $v$ (non-relevant)& Number of characters $f$ on $T$ with $|C|=c$ for which vertex $v$ of $T$ induces monophyly on $T$. \\
\hline
$m_v(T,c)$ & relevant monophyly count of $v$ & Number of characters $f$ on $T$ with $|C|=c$ for which vertex $v$ of $T$ induces relevant monophyly on $T$. \\
\hline
\end{tabular}
\end{table}

\color{black}

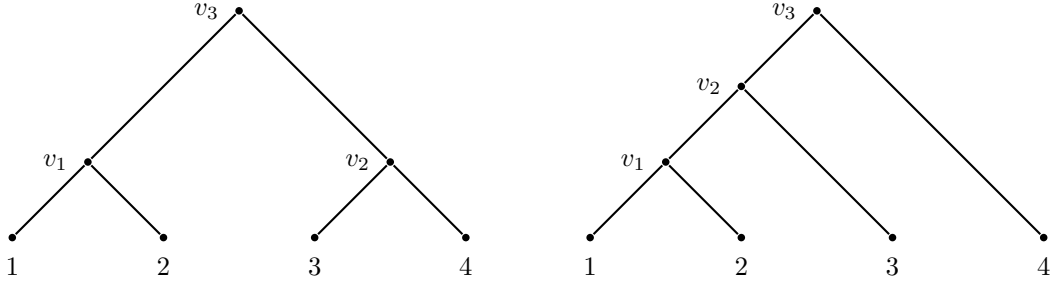
\begin{figure}
    \centering

\begin{tikzpicture}[
    sdot/.style={circle,fill,radius=1pt,inner sep=1pt},
    thick]

\node[sdot] (1) at (0,0) {}; 
\node[sdot] (2) at (2,0) {}; 
\node[sdot] (3) at (4,0) {}; 
\node[sdot] (4) at (6,0) {}; 

\node[sdot] (A) at (1,1) {}; 
\node[sdot] (B) at (5,1) {}; 
\node[sdot] (C) at (3,3) {}; 

\draw (1)--(A)--(2);
\draw (3)--(B)--(4);
\draw (A)--(C)--(B);

\node[below=2pt of 1] {1};
\node[below=2pt of 2] {2};
\node[below=2pt of 3] {3};
\node[below=2pt of 4] {4};

\node[left=2pt of A] {$v_1$};
\node[left=2pt of B] {$v_2$};
\node[left=2pt of C] {$v_3$};
\end{tikzpicture}
\hspace{1cm}
\begin{tikzpicture}[
    sdot/.style={circle,fill,radius=1pt,inner sep=1pt},
    thick]

\node[sdot] (1) at (0,0) {}; 
\node[sdot] (2) at (2,0) {}; 
\node[sdot] (3) at (4,0) {}; 
\node[sdot] (4) at (6,0) {}; 

\node[sdot] (A) at (1,1) {}; 
\node[sdot] (B) at (2,2) {}; 
\node[sdot] (C) at (3,3) {}; 

\draw (1)--(A)--(2);
\draw (3)--(B)--(A);
\draw (4)--(C)--(B);

\node[below=2pt of 1] {1};
\node[below=2pt of 2] {2};
\node[below=2pt of 3] {3};
\node[below=2pt of 4] {4};

\node[left=2pt of A] {$v_1$};
\node[left=2pt of B] {$v_2$};
\node[left=2pt of C] {$v_3$};
\end{tikzpicture}
    \caption{Two trees $T_1,T_2$ on four leaves, with internal vertices labelled.}
    \label{fig:firstexample}
\end{figure}

\begin{example}
    We now calculate the various monophyly numbers by hand for the trees in Figure \ref{fig:firstexample} and list of traits $C = \{\alpha,\beta\}$, that is, for binary characters on each tree. This is shown in Table \ref{tab:mono.nums}. In particular, the non-relevant monophyly numbers and non-relevant full monophyly numbers for both trees are respectively the same, that is $m_N(T_1,2)=m_N(T_2,2)=12$ and $fm_N(T_1,2)=fm_N(T_2,2)=4$. However for $T_1$ the relevant and relevant full monophyly numbers are both $4$, that is $m(T_1,2)=fm(T_1,2)=4$, but for $T_2$ they are $6$ and $2$, respectively. We note here that we will fully characterise all four monophyly numbers for binary characters in Section \ref{s:bincha}.

    \begin{table}[ht]
\centering
\begin{tabular}{|c|cccc|cccc|}
\hline
&
\multicolumn{4}{c|}{\textbf{$T_1$}} &
\multicolumn{4}{c|}{\textbf{$T_2$}} \\
\hline
\textbf{$f$} &
\textbf{M} & \textbf{RM} & \textbf{F} & \textbf{RF} &
\textbf{M} & \textbf{RM} & \textbf{F} & \textbf{RF} \\
\hline
$\alpha \alpha \alpha \alpha$ & \checkmark  & \checkmark  & \checkmark  & \checkmark  & \checkmark  & \checkmark  & \checkmark  & \checkmark  \\
$\alpha \alpha \alpha \beta$  & \checkmark  & $\times$ & $\times$  & $\times$  & \checkmark & \checkmark & \checkmark  & $\times$  \\
$\alpha \alpha \beta \alpha$  & \checkmark  & $\times$  & $\times$  & $\times$  & \checkmark & $\times$  & $\times$  & $\times$  \\
$\alpha \alpha \beta \beta$  & \checkmark  & \checkmark  & \checkmark  & \checkmark & \checkmark  & \checkmark  & $\times$  & $\times$  \\
$\alpha \beta \alpha \alpha$ & \checkmark  & $\times$  & $\times$  & $\times$  & \checkmark & $\times$  & $\times$  & $\times$  \\
$\alpha \beta \alpha \beta$ & $\times$  & $\times$  & $\times$  & $\times$  & $\times$  & $\times$  & $\times$  & $\times$  \\
$\alpha \beta \beta \alpha$ & $\times$  & $\times$  & $\times$  & $\times$  & $\times$  & $\times$  & $\times$  & $\times$  \\
$\alpha \beta \beta \beta$ & \checkmark & $\times$ & $\times$ & $\times$ & \checkmark & $\times$ & $\times$ & $\times$ \\
$\beta \alpha \alpha \alpha$& \checkmark & $\times$ & $\times$ & $\times$ & \checkmark & $\times$ & $\times$ & $\times$ \\
$\beta \alpha \alpha \beta$ & $\times$  & $\times$  & $\times$  & $\times$  & $\times$  & $\times$  & $\times$  & $\times$  \\
$\beta \alpha \beta \alpha$ & $\times$  & $\times$  & $\times$  & $\times$  & $\times$  & $\times$  & $\times$  & $\times$  \\
$\beta \alpha \beta \beta$ & \checkmark  & $\times$  & $\times$  & $\times$  & \checkmark & $\times$  & $\times$  & $\times$  \\
$\beta \beta \alpha \alpha$ & \checkmark & \checkmark & \checkmark & \checkmark & \checkmark & \checkmark & $\times$  & $\times$\\
$\beta \beta \alpha \beta$ & \checkmark  & $\times$  & $\times$  & $\times$  & \checkmark & $\times$  & $\times$  & $\times$  \\
$\beta \beta \beta \alpha$ & \checkmark& $\times$  & $\times$  & $\times$ & \checkmark & \checkmark & \checkmark & $\times$ \\
$\beta \beta \beta \beta$ & \checkmark & \checkmark & \checkmark & \checkmark & \checkmark & \checkmark & \checkmark & \checkmark \\
\hline
\textbf{Total} & 12 & 4 & 4 & 4 & 12 & 6 & 4 & 2 \\
\hline
\end{tabular}
\caption{A table indicating which binary characters on $T_1$ and $T_2$ from Figure \ref{fig:firstexample} are (M)onophyletic, (R)elevant (M)onophyletic, (F)ully monophyletic and (R)elevant (F)ully monophyletic. The totals indicated in the last row correspond to the various monophyly numbers detailed in Table~\ref{tab:mono.nums}. In particular, $m_N(T_1,2)=12$, whereas $m(T_1,2)=fm_N(T_1,2)=fm(T_1,2)=4$. Moreover, $m_N(T_2,2)=12$, $m(T_2,2)=6$, $fm_N(T_2,2)=4$, and $fm(T_2,2)=2$.
}
\label{tab:mono.nums}
\end{table}
\end{example}

\section{A general formula for monophyletic numbers, and some general results}\label{s:genform}
In this section, we derive formulae for the relevant and non-relevant monophyly numbers of a tree $T$ by relating these numbers to counts of certain antichains of vertices in $T$. Although this approach does not appear to extend to the case of relevant or non-relevant full monophyly numbers, we will obtain results for specific cases of these quantities in later sections.

We begin by recalling the standard notion of an antichain in the partial order induced by the ancestor relation on a rooted tree. 

\begin{definition}
    Let $A=\{v_1,\dots,v_k\}$ be a non-empty set of vertices of $T$, such  that, for every distinct $v_i,v_j\in A$, neither $v_i$ is an ancestor of $v_j$ nor $v_j$ is an ancestor of $v_i$.
    Then $A$ is called an \emph{antichain}. If no vertex in $A$ is a leaf, then $A$ is called a \emph{relevant antichain}.
    Furthermore, for an antichain $A$, define $s(A) \coloneqq \vert L(A) \vert$.
\end{definition}

Notice that since the vertices of an antichain are pairwise incomparable under the ancestor relation, the sets $\{L(v):v\in A\}$ are pairwise disjoint, and hence
\[s(A)=|L(A)|=\sum_{v\in A}|L(v)|.\]

We next define the concept of \emph{inducing monophyly}.
\begin{definition}
    Let $T$ be a phylogenetic tree on $n$ leaves, $C$ be a set of traits, $v$ be a vertex of $T$ and $f$ be a character on $T$. Then we say $v$ \emph{induces monophyly on $T$ for $f$} if there is some trait $a \in C$ such that a leaf $\ell$ is assigned trait $a$ by $f$ if and only if $\ell \in L(v)$. If, additionally, $|L(v)| \geq 2$, we say that $v$ \emph{induces relevant monophyly on $T$ for $f$.} Moreover, if $S$ is a set of vertices such that every vertex in $S$ simultaneously induces (relevant) monophyly on $T$ for $f$, then we similarly say that $S$ induces (relevant) monophyly on $T$ for $f$. Finally, if $S$ induces (relevant) monophyly and $f$ is fully monophyletic, we also say that $S$ induces (relevant) full monophyly on $T$ for $f$.
\end{definition}

Let the number of monophyletic characters for which some vertex $v$ induces monophyly and relevant monophyly on $T$ be denoted $m_{N,v}(T,c)$ and $m_v(T,c)$ respectively. We also refer to $m_v(T,c)$ and $m_{N,v}(T,c)$ as the \emph{relevant monophyly count of $v$} and the \emph{monophyly count of $v$}, respectively. As an example, consider vertex $v_2$ of tree $T_2$ in Figure~\ref{fig:firstexample} together with the binary characters listed in Table~\ref{tab:mono.nums}. The characters for which $v_2$ induces monophyly on $T_2$ are $\alpha \alpha \alpha \beta$ and $\beta \beta \beta \alpha$. Both characters are relevant monophyletic and hence, $m_{N,v_2}(T_2,2)=m_{v_2}(T_2,2)=2$. As a second example, consider leaf 1 of $T_2$. Leaf 1 induces monophyly on $T_2$ for the characters $\alpha \beta \beta \beta$ and $\beta \alpha \alpha \alpha$, both of which are not relevant. Thus, $m_{N,1}(T_2,2)=2$, whereas $m_{1}(T_2,2)=0$.

\begin{lemma}
\label{l:single.v.count}
    Let $v$ be an internal vertex or leaf of some phylogenetic tree $T$, and $C$ a set of traits for which $|C|=c$. Then $m_v(T,c)= m_{N,v}(T,c) = c \cdot (c-1)^{n-|L(v)|} $.
\end{lemma}

\begin{proof}
    The leaf descendants of $v$ (or $v$ itself, if $v$ is a leaf) can be assigned any of the $c$ traits in $C$, say $a$, and the remaining leaves of $T$ may be assigned any remaining trait in $C \backslash\{a\}$. 
\end{proof}

\begin{lemma}\label{lem_ancdes}
    Let $T$ be a phylogenetic tree, $v$ be a vertex of $T$, and $f$ be a character on $T$. If $v$ induces monophyly on $T$ for $f$, then no ancestor or descendant of $v$ other than $v$ itself induces monophyly on $T$ for $f$.
\end{lemma}

\begin{proof}
    Suppose $v$ induces monophyly on $T$ for $f$, and let $w$ be an ancestor of $v$ in $T$ for which $w \ne v$. Suppose, seeking a contradiction, that $w$ also induces monophyly on $T$ for $f$. Denote the trait for which $v$ induces monophyly by $a$. Since any descendant of $v$ is a descendant of $w$, the trait for which monophyly is induced by $w$ must also be $a$. Consider some leaf $\ell \in L(w) \backslash L(v)$. Then $\ell$ must also be assigned trait $a$, but $a$ is not a leaf descendant of $v$, contradicting the monophyly induced by $v$. Hence $w$ cannot induce monophyly on $T$ for $f$.

    By the same reasoning, if $w$ were a descendant of $v$ where $w \ne v$, we see $v$ and $w$ cannot both induce monophyly on $T$ for $f$, as $v$ would be an ancestor of $w$. Thus, the statement in the lemma holds.
\end{proof}

The previous observations immediately lead to the following corollary.

\begin{corollary}\label{c:antichain}
    Let $T$ be a phylogenetic tree and $f$ be a character on $T$. Then the set of (internal) vertices that induce (relevant, full, or full relevant) monophyly on $T$ for $f$ form a (relevant) antichain. 
\end{corollary}

We now provide formulae for calculating relevant and non-relevant monophyly numbers. In order to do so, we use the notation $(x)_k = x(x-1)\dots(x-k+1)$, defining $(x)_k=0$ if $x<k$.

\begin{theorem}
\label{thm:gen.theorem}
    Let $T$ be a phylogenetic tree on $n$ leaves and $C$ a set of traits with $|C|=c$. Let $\mathcal{A}$ be the set of antichains that exist in $T$. Then

    \[m_N(T,c) = \sum_{A \in \mathcal{A}} (-1)^{|A|+1} (c)_{|A|} (c-|A|)^{n-s(A)},\]
with the convention that $(c-|A|)^{n-s(A)}=1$ if both $c-|A|=0$ and $n-s(A)=0$. If, instead, $\mathcal{A}$ is the set of relevant antichains, then 

\[m(T,c) = \sum_{A \in \mathcal{A}} (-1)^{|A|+1} (c)_{|A|} (c-|A|)^{n-s(A)},\]

with the same convention.
\end{theorem}

\begin{proof}
The proofs for relevant and non-relevant monophyletic characters proceed almost identically. We simply distinguish between antichains and relevant antichains where appropriate. 

Suppose $f$ is a (relevant) monophyletic character on $T$. Then, at least one vertex in $T$ must induce (relevant) monophyly on $T$ for $f$.

Now, we consider the number of possible ways in which a set $A=\{v_1,\dots,v_k\}$ of (internal) vertices of $T$ can induce monophyly on $T$ for some character $f$, noting that of course, these $k$ (internal) vertices must form a (relevant) antichain by Corollary \ref{c:antichain}, otherwise $A$ cannot induce (relevant) monophyly. Otherwise, we may proceed by assigning any of $c$ traits to the vertices in $L(v_1)$, any of $c-1$ to the vertices in $L(v_2)$, and so on. If $k>c$, then of course again $A$ cannot induce (relevant) monophyly. 

We then consider the number of ways to assign traits to leaves that are not in $\cup_{i=1,\dots,k} L(v_i)$. If there are no such leaves, i.e. $n-s(A)=0$, we have found a set that can simultaneously induce (relevant) monophyly (indeed, we have found a set that induces full (relevant) monophyly).
Otherwise, if there exist such leaves, but we have no possible traits to assign to them, there is no possible way for $A$ to induce (relevant) monophyly. Finally, if there are $n-s(A)>0$ such leaves and sufficient traits to assign to them, we may assign any of the $c-|A|$ traits to each of them, in $(c-|A|)^{n-s(A)}$ possible ways.

It follows, for a given (relevant) antichain $A$, that there are 
\[ (c)_{|A|} (c-|A|)^{n-s(A)}\]
possible characters for which $A$ induces (relevant) monophyly on $T$ (with the convention that $(c-|A|)^{n-s(A)}=1$ if both $c-|A|=0$ and $n-s(A)=0$, and noting that if $c-|A|=0$ in any other case we have seen that $A$ cannot induce (relevant) monophyly on $T$).

Finally, the expression in the theorem statement is obtained by application of the Inclusion-Exclusion Principle.
\end{proof}

Suppose that instead we wish to count the number of (relevant) monophyletic characters that induce at least $k$ monophyletic traits. Denote the number of relevant monophyletic characters on $T$ with $c$ traits and at least $k$ relevant monophyletic traits by $m(T,c,k)$. Define the number for non-relevant monophyletic characters similarly and denote it $m_N(T,c,k)$. Then we have the following theorem, with proof omitted as it proceeds almost identically to the previous proof. Note in particular the modified exponent of $(-1)$, due to the application of the Inclusion-Exclusion Principle starting at antichains of size $k$.

\begin{theorem}
Let $T$ be a phylogenetic tree on $n$ leaves, and let $C$ be a set of traits with $|C|=c$. Let $\mathcal{A}_k$ denote the set of antichains of size at least $k$ that exist in $T$. Then
\[
m_N(T,c,k)
=
\sum_{A\in \mathcal{A}_k}
(-1)^{|A|+k}
(c)_{|A|}
(c-|A|)^{n-s(A)},
\]
with the convention that
\[
(c-|A|)^{n-s(A)}=1
\]
whenever both $c-|A|=0$ and $n-s(A)=0$. If, instead, $\mathcal{A}$ is the set of relevant antichains of size at least $k$, then

\[
m(T,c,k)
=
\sum_{A\in \mathcal{A}_k}
(-1)^{|A|+k}
(c)_{|A|}
(c-|A|)^{n-s(A)},
\]
with the same convention.
\end{theorem}

\begin{example}
    Let $T$ be the tree on $6$ leaves depicted in Figure \ref{fig:sixleavesexample}. We will show that $m(T,3)=141$ (noting of course that this is the \emph{relevant} monophyly number). The antichains of size $1$ in $T$ are simply $A_i=\{v_i\}$ for $i=1,\dots,5$. The antichains of size $2$ are $A_6=\{v_1,v_2\}$, $A_7=\{v_1,v_3\}$, $A_8=\{v_2,v_3\}$, and $A_9= \{v_3,v_4\}$. Finally, the single antichain of size $3$ is $A_{10}=\{v_1,v_2,v_3\}$, and there are no larger antichains. For each of $A_1,A_2$, and $A_3$, $T$ has four leaves that are not descendants of their vertices, and hence they contribute $3\cdot2^4=48$ characters. For $A_4$, there are two leaves that are not descendants, so $3\cdot2^2=12$ characters are contributed, and for $A_5$ there are no leaves that are not descendants, so $3$ characters are contributed. This is a total of $159$ monophyletic characters on $T$ induced by the antichains of size $1$.

    Each of $A_6,A_7$, and $A_8$ contain $2$ vertices, there are two leaves in $T$ that are not descendants of one of their vertices, and hence they contribute $(3)_2 \cdot 1^2 =6$ monophyletic characters. For $A_9$ there are no leaves that are not descendants of one of its vertices, so $A_9$ contributes $(3)_2=6$ as well. Hence there are $24$ monophyletic characters on $T$ induced by the antichains of size $2$.

    Finally, for the single antichain of size $3$, $T$ has no leaves that are not descendants of one of its vertices, so there are $(3)_3=6$ monophyletic characters induced by $A_{10}$. Thus, by applying the Inclusion-Exclusion Principle, we have $m(T,3)=159-24+6=141$. 
\end{example}

\begin{figure}
    \centering

\begin{tikzpicture}[
    sdot/.style={circle,fill,radius=1pt,inner sep=1pt},
    thick
]

\node[sdot] (1) at (0,0) {}; 
\node[sdot] (2) at (2,0) {}; 
\node[sdot] (3) at (4,0) {}; 
\node[sdot] (4) at (6,0) {}; 
\node[sdot] (5) at (8,0) {}; 
\node[sdot] (6) at (10,0) {}; 

\node[sdot] (A) at (1,1) {}; 
\node[sdot] (B) at (5,1) {}; 
\node[sdot] (C) at (3,3) {}; 
\node[sdot] (D) at (9,1) {}; 
\node[sdot] (E) at (5,5) {}; 

\draw (1)--(A)--(2);
\draw (3)--(B)--(4);
\draw (A)--(C)--(B);
\draw (5)--(D)--(6);
\draw (C)--(E)--(D);

\node[below=2pt of 1] {1};
\node[below=2pt of 2] {2};
\node[below=2pt of 3] {3};
\node[below=2pt of 4] {4};
\node[below=2pt of 5] {5};
\node[below=2pt of 6] {6};

\node[left=2pt of A] {$v_1$};
\node[left=2pt of B] {$v_2$};
\node[left=2pt of C] {$v_4$};
\node[left=2pt of D] {$v_3$};
\node[left=2pt of E] {$v_5$};
\end{tikzpicture}
    \caption{A tree $T$ on six leaves, with internal vertices labelled.}
    \label{fig:sixleavesexample}
\end{figure}
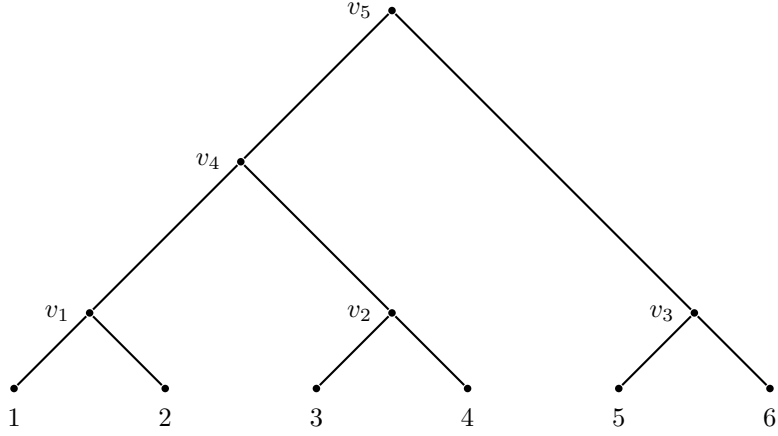

\section{Enumerating monophyletic characters based on the underlying tree shape}\label{s:treeshape} 

It can easily be seen that the number of (relevant/non-relevant and/or fully/non-fully) monophyletic characters employing $c$ different traits on a given tree $T$ does not depend on the actual leaf labeling of $T$, but only on its tree shape.  
So, in this subsection, we investigate the dependency of the various monophyly numbers on the underlying tree shape. In particular, we will investigate caterpillar trees and  fully balanced trees, as these represent extremal trees in terms of balance. 

\subsection{Calculating  the monophyly numbers for caterpillars}

\begin{lemma}\label{lem:cat_anti}
    Let $T^{cat}_n$ be a caterpillar tree with $n \in \mathbb{N}_{\ge 2}$ leaves. Then every antichain of $T_n^{cat}$ either consists entirely of leaves, or consists of a single internal vertex $v$ together with some subset of the leaves not descended from $v$.
\end{lemma}

\begin{proof}
    The proof follows directly from the fact that the internal vertices of $T_n^{cat}$ form a chain (that is, they are totally ordered by the ancestor-descendant relation on $T_n^{cat}$).
\end{proof}

\begin{theorem}\label{thm_catvalue}
Let $T_n^{cat}$ be a caterpillar tree with $n\in\mathbb{N}_{\geq 2}$ leaves. Let $c\in\mathbb{N}_{\geq 2}$ be a number of traits. Then, we have:  $m(T_n^{cat},c)=\sum\limits_{j=0}^{n-2}c\cdot(c-1)^{j}$.
\end{theorem}

\begin{proof}
By Lemma \ref{lem:cat_anti}, every relevant antichain in $T_n^{cat}$ consists of a single internal vertex. Moreover, for each $j\in\{2,\dots,n\}$, there is exactly one internal vertex with $j$ descendant leaves.

By Lemma \ref{l:single.v.count}, an internal vertex with $j$ descendant leaves contributes
\[
c(c-1)^{n-j}
\]
relevant characters. Summing over all possible $j$ gives
\[
m(T_n^{cat},c)
=
\sum_{j=2}^{n} c(c-1)^{n-j}
=
\sum_{j=0}^{n-2} c(c-1)^j.
\]
\end{proof}

\begin{theorem}\label{thm_catvalue_nonrelevant}
Let $T_n^{cat}$ be a caterpillar tree with $n\in\mathbb{N}_{\geq 2}$ leaves. Then, we have:  \[
m_N(T_n^{cat},c)
=
\sum_{k=1}^{n}
(-1)^{k+1}
(c)_k
\binom{n}{k}
(c-k)^{n-k}
\;+\;
\sum_{j=2}^{n}
\sum_{k=1}^{n-j+1}
(-1)^{k+1}
(c)_k
\binom{n-j}{k-1}
(c-k)^{n-j}.
\]
\end{theorem}

\begin{proof}
By Lemma \ref{lem:cat_anti}, the antichains of $T_n^{cat}$ either consist only of leaves, or one internal vertex $v$ together with some subset of the leaves that are not descendants of $v$. We first consider the antichains consisting only of leaves. 

Finding antichains of length $k$ consisting only of leaves amounts to selecting $k$ of the $n$ leaves. Applying Theorem \ref{thm:gen.theorem} directly, we then obtain that the contribution of these antichains is 
\begin{equation}\label{e:sum1} 
\sum_{k=1}^{n}
(-1)^{k+1}
(c)_k
\binom{n}{k}
(c-k)^{n-k}.
\end{equation}

We now consider the antichains consisting of one internal vertex $v$ together with some subset of the leaves that are not descendants of $v$. In this case, for each of $j=2,\dots,n$, there exists an internal vertex $v$ with $j$ descendant leaves, and the antichains of length $k$ can be found by simply selecting $v$ and then choosing $k-1$ of the $n-j$ remaining leaves. Applying Theorem \ref{thm:gen.theorem}, summing over possible $k$, then summing over possible $j$ gives
\begin{equation}\label{e:sum2}
\sum_{j=2}^{n}
\sum_{k=1}^{n-j+1}
(-1)^{k+1}
(c)_k
\binom{n-j}{k-1}
(c-k)^{n-j}.
\end{equation}

Summing expressions \eqref{e:sum1} and \eqref{e:sum2} together gives the result in the theorem.
\end{proof}

We now consider relevant full monophyly. We prove a more general result that has an immediate corollary for caterpillars.

\begin{theorem}
    \label{thm:full.mono}
Let $T=(T_1,T_2)$ be a rooted binary phylogenetic tree with $n\geq 2$ leaves and with maximal pendant subtrees $T_1$ and $T_2$ with $n_1$ and $n_2$ leaves, respectively, such that $n_1 \geq n_2$. If $n_2=1$, then $fm(T,c)=c$. 
\end{theorem}

\begin{proof}
    Suppose $n_2 =1$ and $f$ is a relevant fully monophyletic character on $T$. As $f$ is relevant fully monophyletic, the leaf child of the root, say $\ell$, must be contained in a monophyletic group of size at least two. Hence, there is some inner vertex $v$ that is an ancestor of $\ell$, and all descendants of $v$ share the same trait. As the root is the only inner vertex of $T$ that is an ancestor of $\ell$, it follows that all descendants of the root, in particular all leaves, must have the same trait. As there are $c$ possible traits, $fm(T,c)=c$.
\end{proof}

\begin{corollary}
    Let $T_n^{cat}$ be a caterpillar tree with $n \geq 2$ leaves. Then $fm(T_n^{cat},c)=c$ (independent of $n$). 
\end{corollary}

\begin{proof}
The caterpillar tree $T_n^{cat}$ has a leaf child of the root and hence falls under the assumptions of Theorem \ref{thm:full.mono}.
\end{proof}

\begin{theorem}
    Let $T^{cat}_n$ be a caterpillar tree with $n \ge 2$ leaves. Then we have 

    \[ fm_N(T^{cat}_n,c) = \sum_{k=1}^n (c)_{n-k+1}. 
\]
\end{theorem}

\begin{proof}
    Lemma \ref{lem:cat_anti} states that the antichains of $T_n^{cat}$ either consist only of leaves, or one internal vertex $v$ together with some subset of the leaves that are not descendants of $v$. It follows that every fully monophyletic character is induced either by
    \begin{enumerate}
        \item an antichain consisting of some internal vertex $v$ together with every leaf that is not a descendant of $v$, or
        \item an antichain consisting entirely of leaves.
    \end{enumerate}

    Suppose $v$ has $k$ descendant leaves. Then all descendants of $v$ must be assigned a common trait, while every remaining leaf must be assigned a distinct trait different from that common trait. Hence the number of possible characters induced by such a vertex is
    \[
    c(c-1)_{n-k} = (c)_{n-k+1}.
    \]

    Summing over all internal vertices therefore gives
    \[
    \sum_{k=2}^{n} (c)_{n-k+1},
    \]
    since in a rooted binary caterpillar there is exactly one internal vertex with $k$ descendant leaves for each $k\in\{2,\dots,n\}$.

    It remains to consider antichains consisting entirely of leaves. Such an antichain induces full monophyly only if every leaf is assigned a distinct trait, giving
    \[
    (c)_n
    \]
    possible characters.

    Hence 
    \begin{align*}
    fm_N(T^{cat}_n,c)
    & =
    \sum_{k=2}^{n} (c)_{n-k+1}
    +
    (c)_n \\
    &= \sum_{k=1}^n (c)_{n-k+1}.    
    \end{align*}
    
    As usual, we adopt the convention that $(x)_k=0$ whenever $x<k$. This completes the proof. 
\end{proof}

\subsection{Calculating the monophyly numbers for fully balanced trees}

In the following, given a phylogenetic tree $T$, we refer to a (relevant) antichain $A$ of size $k$ as a \emph{(relevant) $k$-antichain}. Furthermore, we say that such a  (relevant) $k$-antichain $A$ \emph{covers $m$ leaves} if $s(A) = m$. Note that this notion depends only on the topology of the tree and not on the labeling of the leaves.

We begin by counting (relevant) $k$-antichains in $T_h^{fb}$ that cover exactly $m$ leaves. Both types satisfy the same recursion, differing only in their base cases: relevant $k$-antichains consist solely of internal vertices, whereas general $k$-antichains may also contain leaves. In the following statements, $\boldsymbol{1}_{\{\text{conditions}\}}$ denotes the indicator function, which takes the value $1$ if the conditions are met, and is 0 otherwise.

\begin{lemma} \label{lem:antichaincover-fb}
    Let $h,k \in \mathbb{N}_{\geq 0}$. Let $a(h,k,m)$ denote the number of $k$-antichains in $T_h^{fb}$ that cover exactly $m$ leaves, and let $b(h,k,m)$ denote the number of relevant $k$-antichains in $T_h^{fb}$ that cover exactly $m$ leaves. Set
    \[ a(h,k,m)= b(h,k,m)=0  \text{ whenever }  m>2^h.\] 
    Then the following recursions hold:
    \begin{enumerate}
        \item For $a(h,k,m)$, we have
           \[ a(0,0,0)=1, \quad a(0,1,1) = 1, \quad \text{and } a(0,k,m)=0 \text{ otherwise}.\] For $h \geq 1$, 
         \[ a(h,k,m) = \mathbf{1}_{\{k=1,m=2^h\}} + \sum\limits_{i=0}^{k} \sum\limits_{m_1=0}^{2^{h-1}} a(h-1,i,m_1) \cdot a(h-1,k-i,m-m_1).\]
        \item For $b(h,k,m)$, we have
          \[ b(0,0,0)=1, \quad \text{and } b(0,k,m) = 0 \text{ for } (k,m) \neq (0,0).\] For $h \geq 1$, 
         \[ b(h,k,m) = \mathbf{1}_{\{k=1,m=2^h\}} + \sum\limits_{i=0}^{k} \sum\limits_{m_1=0}^{2^{h-1}} b(h-1,i,m_1) \cdot b(h-1,k-i,m-m_1).\]
    \end{enumerate}
\end{lemma}

Before we prove the lemma, we present an example. 

\begin{example}\label{ex:fully-balanced-antichains}
    Consider $T_3^{fb}$, the fully balanced tree of height three. Then, we obtain the following non-zero counts for $a(3,k,m)$ and corresponding counts for $b(3,k,m)$, listing only those for $1 \leq k \leq 3$:
    \begin{center}
        \begin{tabular}{c|c|c|c}
          $k$ & $m$ & $a(3,k,m)$ & $b(3,k,m)$ \\ \hline
          1 & 1 & 8 & 0\\
          1 & 2 & 4 & 4\\
          1 & 4 & 2 & 2 \\
          1 & 8 & 1 & 1 \\
          2 & 2 & 28 & 0\\
          2 & 3 & 24 & 0\\
          2 & 4 & 6 & 6 \\
          2 & 5 & 8 & 0\\
          2 & 6 & 4 & 4 \\
          2 & 8 & 1 & 1 \\
          3 & 3 & 56 & 0 \\
          3 & 4 & 60 & 0\\
          3 & 5 & 24 & 0 \\
          3 & 6 & 16 & 4 \\
          3 & 7 & 8 & 0\\
          3 & 8 & 2 & 2 \\
        \end{tabular}
    \end{center}
\end{example}

Now we are in a position to prove Lemma~\ref{lem:antichaincover-fb}.

\begin{proof}[Proof of Lemma~\ref{lem:antichaincover-fb}]
We prove the recursion by induction on $h$. The argument is identical for $a(h,k,m)$ and $b(h,k,m)$, except for the differing base cases.

For $h=0$, the statement follows directly from the definitions. In particular, $T_0^{fb}$ consists of a single leaf, so the only admissible antichains are:
    \begin{itemize}
        \item the empty antichain, yielding $a(0,0,0)=b(0,0,0)=1$, and
        \item in the case of $a(h,k,m)$, the single-element antichain consisting of the leaf, yielding $a(0,1,1)=1$.
    \end{itemize}

Now assume the recursion holds for height $h-1$, and consider $T_h^{fb}$ with root $r$. The two maximal pendant subtrees of $T_h^{fb}$ are copies of $T_{h-1}^{fb}$.

Let $A$ be a $k$-element (relevant) antichain in $T_h^{fb}$ that covers exactly $m$ leaves. Two mutually exclusive cases can occur.
\begin{itemize}
    \item If $\rho \in A$, then no descendant of $\rho$ can belong to the (relevant) antichain. Hence $A = \{\rho\}$. The root covers all leaves of the tree, which are exactly $2^h$. Thus this case occurs only when $k=1$ and $m=2^h$. Therefore it contributes $\mathbf{1}_{\{k=1,m=2^h\}}$.

    \item If $\rho \notin A$, then every element of $A$ lies in one of the two maximal pendant subtrees of $T_h^{fb}$. Since these subtrees are disjoint, the (relevant) antichain splits uniquely into $i$ inner vertices in one subtree, say the left subtree, and $k-i$ inner vertices in the right subtree for some $0 \leq i \leq k$. Let $m_1$ be the number of leaves covered by the vertices in the left subtree, and let $m-m_1$ be the number of leaves covered by the vertices in the right subtree. 
    
    By the inductive hypothesis, there are  
    \[ d(h-1,i,m_1) \cdot d(h-1,k-i,m-m_1)\]
    possibilities (where $d(\cdot, \cdot, \cdot)$ denotes either $a(\cdot, \cdot, \cdot)$ or $b(\cdot, \cdot, \cdot)$).

    Summing over all $i$ and $m_1$ gives
    \[ \sum\limits_{i=0}^{k} \sum\limits_{m_1=0}^{2^{h-1}} d(h-1,i,m_1) \cdot d(h-1,k-i,m-m_1)\]
\end{itemize}
Combining the two cases gives the claimed recursion.
\end{proof}

We now use the preceding lemma and Theorem~\ref{thm:gen.theorem} to count relevant and non-relevant monophyly numbers for $T_h^{fb}$.

\begin{theorem}\label{thm_fbvalue}
Let $T_h^{fb}$ be a fully balanced tree of height $h \geq 1$. Let $c\in\mathbb{N}_{\geq 2}$ be a number of traits. Then, we have 
\begin{align*}
    m_N(T_h^{fb},c) &= \sum\limits_{k = 1}^{2^h} \sum\limits_{m =1}^{2^h} (-1)^{k+1} \, (c)_k \, a(h,k,m) \, (c-k)^{2^h-m} \quad \text{and}\\
    m(T_h^{fb},c)&=\sum\limits_{k =1}^{2^h} \sum\limits_{m =1}^{2^h} (-1)^{k+1} \, (c)_k \, b(h,k,m) \, (c-k)^{2^h-m}
\end{align*}
where $a(h,k,m)$ and $b(h,k,m)$ are computed as in Lemma~\ref{lem:antichaincover-fb}, and where we use the convention that $(-1)^{k+1} \, (c)_k \, a(h,k,m) \, (c-k)^{2^h-m} = 1 = (-1)^{k+1} \, (c)_k \, b(h,k,m) \, (c-k)^{2^h-m} = 1$ whenever both $c=k$ and $2^h=m$.
\end{theorem}

\begin{proof}
Let $\mathcal{A}$ denote the set of antichains that exist in $T_h^{fb}$.
Recall from Theorem~\ref{thm:gen.theorem} that 
\[m_N(T_h^{fb},c) = \sum_{A \in \mathcal{A}} (-1)^{|A|+1} (c)_{|A|} (c-|A|)^{n-s(A)},\]
with the convention that $(c-|A|)^{n-s(A)}=1$ whenever $c-|A|=0$ and $n-s(A)=0$. 

Grouping the antichains by their size $k = |A|$ and their covered-leaf count $m = s(A)$ and noting that $n=2^h$, yields
\begin{align*}
    m_N(T_h^{fb},c) &= \sum\limits_{k =1}^{2^h} \sum\limits_{m =1}^{2^h} (-1)^{k+1} \, (c)_k \, a(h,k,m) \, (c-k)^{n-m},
\end{align*}
where $a(h,k,m)$ is defined as in Lemma~\ref{lem:antichaincover-fb}.

If $\mathcal{A}$ instead denotes the set of relevant antichains that exist in $T_h^{fb}$, the same argument applies with $a(h,k,m)$ replaced by $b(h,k,m)$. This yields
\begin{align*}
    m(T_h^{fb},c) &= \sum\limits_{k = 1}^{2^h} \sum\limits_{m = 1}^{2^h} (-1)^{k+1} \, (c)_k \, b(h,k,m) \, (c-k)^{2^h-m},
\end{align*}
where $b(h,k,m)$ is defined as in Lemma~\ref{lem:antichaincover-fb}.

\end{proof}

As Theorem~\ref{thm:gen.theorem} only covers (relevant) monophyletic characters but not fully (relevant) monophyletic characters, we cannot directly apply it to determine the remaining two monophyly counts, $fm(T_h^{fb},c)$ and $fm_N(T_h^{fb},c)$.

In the following, we first establish a formula for $fm(T_{h}^{fb},c)$, and then relate $fm(T_{h}^{fb},c)$ and $fm_N(T_{h-1}^{fb},c)$, which, taken together, allows us to also determine $fm_N(T_{h}^{fb},c)$.

\medskip
Recall that $b(h,k,2^h)$ (cf. Lemma \ref{lem:antichaincover-fb}) is the number of relevant $k$-antichains covering all leaves in $T^{fb}_h$. 

\begin{theorem}\label{thm:fully-balanced-fully-relevant}
    Let $T_h^{fb}$ be a fully balanced tree of height $h \geq 1$. Let $c \in \mathbb{N}_{\geq 2}$ be a number of traits. Then, we have
    \[ fm(T_h^{fb},c) =  \sum\limits_{k=1}^{2^{h-1}} b(h,k,2^h) \cdot (c)_k,\]
    where $b(h,k,2^h)$ is computed as in Lemma~\ref{lem:antichaincover-fb}.  
\end{theorem}

\begin{proof}
Let $A = \{v_1, \ldots, v_k\}$ be a relevant $k$-antichain with $s(A) = V^1(T^{fb}_h)$. If $C$ is a set of traits with $|C|=c$ and $k \leq c$, then we may first choose an injective assignment of traits to the elements of $A$; that is, we select $k$ distinct traits from $C$ and assign one each to $v_i$. This can be done in $(c)_k$ ways.

Having fixed such an assignment, we then assign each leaf in $L(v_i)$ the trait that is assigned to $v_i$. Since $A$ is a relevant antichain covering all leaves, this yields a relevant fully monophyletic character on $T^{fb}_h$. By construction, this character uses exactly $k$ distinct states. 

Conversely, every relevant fully monophyletic character arises uniquely from such a pair $(A,\phi)$, where $A$ is a relevant $k$-antichain covering all leaves of $T^{fb}_h$ and $\phi$ is an injective assignment of $k$ traits from $C$ to the elements of $A$.

Therefore, the number of relevant fully monophyletic characters is given by
\[ fm(T_h^{fb},c) = \sum\limits_{k=1}^{2^{h-1}} b(h,k,2^h) \cdot (c)_k.\]
Notice that the upper bound in the sum follows from the fact that every relevant $k$-antichain that covers all leaves of $T^{fb}_h$ corresponds to a partition of the leaves into disjoint monophyletic groups, each of which must contain at least two leaves. Hence, $k \leq 2^{h-1}$.
\end{proof}

We now relate $fm(T_{h}^{fb},c)$ and $fm_N(T_{h-1}^{fb},c)$.

\begin{proposition} \label{prop:fb-fully-counts}
    Let $T_{h-1}^{fb}$ and $T_{h}^{fb}$ be fully balanced trees of heights $h-1$ and $h$, respectively, where $h \geq 2$. Let $c \in \mathbb{N}_{\geq 2}$ be a number of traits. Then,
    \[ fm(T_{h}^{fb},c) = fm_N(T_{h-1}^{fb},c).\]
\end{proposition}

\begin{proof}
    We prove this statement by establishing a bijection between fully monophyletic characters on $T_{h-1}^{fb}$ and relevant fully monophyletic characters for $T_{h}^{fb}$. 

    First, let $f$ be a fully monophyletic character on $T_{h-1}^{fb}$. Construct $T_h^{fb}$ by replacing each leaf of $T_{h-1}^{fb}$ with a cherry. Define a character $f'$ on $T_h^{fb}$ by assigning both leaves of each cherry the same trait as the original leaf. Since each leaf of $T_{h-1}^{fb}$ is replaced by two leaves with identical traits, every monophyletic group in $f'$ has size at least two. Moreover, the monophyletic structure is preserved under this transformation. Hence, $f'$ is a relevant fully monophyletic character on $T_h^{fb}$.

    Conversely, let $f$ be a relevant fully monophyletic character on $T_h^{fb}$. Then every leaf belongs to a monophyletic group of size at least two. In particular, if $[x,y]$ is a cherry of $T_h^{fb}$, then $f(x)=f(y)$. Now contract each cherry $[x,y]$ into a single leaf, denoted $xy$, thereby obtaining $T_{h-1}^{fb}$. Define a character $f'$ on $T_{h-1}^{fb}$ by setting $f'(xy)=f(x)=f(y)$. Since $f$ was fully monophyletic and the contraction preserves the induced subtree structure of each trait class, it follows that $f'$ is a fully monophyletic character on $T_{h-1}^{fb}$.

    It is straightforward to verify that these two constructions are inverse to each other. Therefore, we obtain a bijection between fully monophyletic characters on $T_{h-1}^{fb}$ and relevant fully monophyletic characters on $T_h^{fb}$, which completes the proof.
\end{proof}

A direct consequence of Theorem~\ref{thm:fully-balanced-fully-relevant} and Proposition~\ref{prop:fb-fully-counts} is the following corollary.

\begin{corollary}
    Let $T_h^{fb}$ be a fully balanced tree of height $h \geq 1$. Let $c \in \mathbb{N}_{\geq 2}$ be a number of traits. Then, we have
    \[ fm_N(T_h^{fb},c) = fm(T_{h+1}^{fb},c) = \sum\limits_{k=1}^{2^{h}} b(h+1,k,2^{h+1}) \cdot (c)_k,\]
    where $b(h+1,k,2^{h+1})$ is computed as in Lemma~\ref{lem:antichaincover-fb}.  
\end{corollary}

\section{Calculating the monophyly numbers for binary characters}\label{s:bincha}
In this section, we calculate the monophyly numbers for general tree shapes with binary characters. Restricting to binary characters allows substantially sharper structural characterisations than in the general $c$-state setting.

We start with a simple corollary, which is a direct consequence of Theorem \ref{thm_catvalue}.

\begin{corollary} \label{cor_valuecat} Let $T_n^{cat}$ be a caterpillar tree with $n\geq 2$ leaves. Then, we have:  $m(T_n^{cat},2)=2n-2$.  \end{corollary}

\begin{proof} By Theorem \ref{thm_catvalue}, we have: $$m(T_n^{cat},2)=\sum\limits_{j=0}^{n-2}2\cdot(2-1)^{j}= \sum\limits_{j=0}^{n-2}2\cdot 1^{j}=2(n-1)=2n-2.$$ This completes the proof. \end{proof}

While Corollary \ref{cor_valuecat} provides an explicit value for $m(T_n^{cat},2)$, it does not take into account any other tree shapes. We will shortly prove a theorem that is a lot stronger. It shows that the value of the four types of monophyly numbers for $T_n^{cat}$ do not depend on any other property of the caterpillar than on the fact that the root has a leaf child, as all trees with this property share the same monophyly numbers with $T_n^{cat}$. Indeed, the non-relevant (full) monophyly number depends only on the number of leaves in the tree. However, first we must prove a lemma characterising certain antichains.

\begin{lemma}\label{lem_len2anti}
    Let $T$ be a rooted binary phylogenetic tree with $n\geq 2$ leaves. Then there is a unique antichain of size $2$ that covers all the leaves of $T$, consisting of the two child vertices of the root of $T$.
\end{lemma}

\begin{proof}
 Denote the two maximal pendant subtrees of $T$ by $T_1$ and $T_2$. We now show that the only way to cover the leaves of $T$ using two vertices is by selecting $v_1$ and $v_2$ to be the roots of $T_1$ and $T_2$ in some order (in which there are of course two ways to do this). One can immediately observe that this selection covers the leaves of $T$, so all that remains is to show no other pair of vertices is possible. Neither vertex can be the root, as otherwise we could not have a size-$2$ antichain by Lemma \ref{lem_ancdes}. Now, seeking a contradiction, and without loss of generality, suppose $v_1$ was any other internal vertex aside from the roots of $T,T_1$ or $T_2$, but $v_1,v_2$ still cover the leaves of $T$. Then the maximal pendant subtree containing $v_1$ must have at least one leaf $\ell_1$ that is not a descendant of $v_1$ (as $v_1$ is not a root of either maximal pendant subtree). But any leaf $\ell_2$ selected from the other maximal pendant subtree is also not a descendant of $v_1$. This is a contradiction, as this would require $v_2$ to cover leaves in both maximal pendant subtrees and hence be the root, which we have already shown is not possible.

It follows that there is only one antichain of size $2$ that covers all leaves in $T$.
\end{proof}

By leveraging Lemma \ref{lem_len2anti}, we can completely determine all four types of monophyly counts for binary characters. Interestingly, for non-relevant (fully) monophyletic characters, the count depends only on the number of leaves, and for relevant (fully) monophyletic characters, it depends only on the number of leaves and whether the root has a leaf child.

\begin{theorem}\label{thm:binary.summary}
Let \(T=(T_1,T_2)\) be a rooted binary phylogenetic tree with
\(n\ge 2\) leaves, where \(T_1,T_2\) have \(n_1,n_2\) leaves respectively and \(n_1\ge n_2\). Then:

\begin{align*}
m(T,2) &=
\begin{cases}
2n-4,& n_2>1,\\
2n-2,& n_2=1,
\end{cases} \\
m_N(T,2)&=4n-4,\\
fm(T,2) &=
\begin{cases}
4,& n_2>1,\\
2,& n_2=1,
\end{cases} \qquad \text{and}\\
fm_N(T,2)&=4.
\end{align*}
\end{theorem}

\begin{proof}
By Theorem \ref{thm:gen.theorem} and the convention $(x)_k=0$ for $x<k$,
only antichains of size at most $2$ contribute when $c=2$.

Moreover, if $|A|=2$, then the term
\[
(c-|A|)^{n-s(A)}
\]
is non-zero only when $s(A)=n$, that is, when the antichain covers all leaves of $T$.

By Lemma \ref{lem_len2anti}, there is a unique antichain of size $2$ that covers all leaves, namely the set consisting of the two children of the root.

Now, each type of monophyly count differs only by which antichains are admissible

\begin{itemize}
\item non-relevant monophyly: all antichains;
\item relevant monophyly: relevant antichains, i.e., antichains that contain no leaves;
\item full monophyly: antichains that cover all leaves;
\item relevant full monophyly: relevant antichains that cover all leaves.
\end{itemize}

We first consider non-relevant monophyly. For non-relevant monophyly, every vertex of $T$ may form an antichain of size $1$. Since a rooted binary tree with $n$ leaves has $2n-1$ vertices, these antichains contribute
\[
2(2n-1)=4n-2.
\]
The unique covering antichain of size $2$ contributes
\[
-(2)_2= -2,
\]
and hence
\[
m_N(T,2)=4n-4.\]

We now consider non-relevant full monophyly. In this case, we must restrict the size-$1$ antichains to also specifically be those that cover all leaves of the tree, which means the only permitted size-$1$ antichain consists of the root. There are two possible trait assignments for descendants of the root, which together with the two possible assignments for the antichain of size $2$ from Lemma \ref{lem_len2anti} gives us that $fm_N(T,2)=4$.

For relevant monophyly, only internal vertices may appear in antichains. Since $T$ has $n-1$ internal vertices, antichains of size $1$ contribute
\[
2(n-1)=2n-2.
\] 

Again, by Lemma \ref{lem_len2anti}, there is only one possible size-$2$ antichain. However, if $n_2=1$, one of the vertices in this antichain will be a leaf, so this is not a permissible antichain. Hence if $n_2=1$, $m(T,2)=2n-2$. However, if $n_2>1$, it is permissible, and so by Theorem \ref{thm:gen.theorem}, we must subtract the two possible trait assignments, and in this case $m(T,2)=2n-4$.

Finally, we consider relevant full monophyly. In this case we are restricted to only antichains that cover the leaves of $T$, and furthermore do not contain leaves. There are therefore two possible antichains -- the size-$1$ antichain containing only the root, and the size-$2$ antichain consisting of the two children of the root. If $n_2=1$, then the latter contains a leaf, and hence the only valid relevant fully monophyletic characters are those that assign all descendants of the root the same character, of which there are $2$. Hence if $n_2=1$, we obtain $fm(T,2)=2$. If, however, $n_2>1$, the size-$2$ antichain is valid, with a further two assignments, and thus $fm(T,2)=4$.

This proves the theorem.
\end{proof}

In the following section, we use our combinatorial results to analyse the expected values of the various monophyletic character counts under different models. 

\section{Expectations of monophyly}\label{s:expmon}

Phylogenetic trees can be generated under different models. Thus far, we have characterised all four monophyly numbers for trees on binary characters. We will now consider their expected values under two simple models. The first one is the \emph{uniform model}, which assigns the same probability $p$ to all rooted binary phylogenetic trees on $n$ leaves. As there are $(2n-3)!!$ many such trees \cite{Semple2003}, this gives $p=\frac{1}{(2n-3)!!}$. 

The other model we consider is the \emph{Yule model}. This model starts with a single vertex, and repeats the following step $n-1$ times: from all leaves currently present in the tree, choose one uniformly at random and replace it by a cherry. Finally, leaf labels are assigned to the resulting tree uniformly at random from the set of labels, e.g., $X=\{1,\ldots,n\}$. 

Theorem \ref{thm:binary.summary} gives us a simple way to calculate the expected values of $m(T,2)$ and $fm(T,2)$ both under the uniform and the Yule model. We do not consider $m_N(T,2)$ or $fm_N(T,2)$, as they are constant for fixed $n$. We start with the uniform model.

\begin{proposition}\label{cor_binaryuniform} 
Let $n \geq 2$ and let $T_n$ be a phylogenetic tree with $n$ leaves sampled under the uniform model. Then, 
$$E_{unif}[m(T_n,2)]=2n-3+\frac{3}{2n-3}.$$
and 
$$E_{unif}[fm(T_n,2)]=3 - \frac{3}{2n-3}.$$
\end{proposition}

\begin{proof}
There are $(2n-3)!!$ many rooted binary phylogenetic $X$-trees \cite{Semple2003}. Of these trees, there are $n\cdot (2n-5)!!$ many that have a singleton leaf adjacent to the root. This can be seen by choosing one of the $n$ leaves for the role of this singleton leaf and choosing one of the  $(2(n-1)-3)!!=(2n-5)!!$ many trees on the remaining taxa as the second maximal pendant subtree. In the following, whenever a tree is such that one of its leaves is adjacent to the root, we say it has property $\pi$. So the probability that $T_n$ has property $\pi$ is $P(T_n:\pi) = \frac{n\cdot (2n-5)!!}{(2n-3)!!}=\frac{n}{2n-3} $.
Moreover, for all trees $T$ that have property $\pi$, we have $m(T,2)=2n-2$ and $fm(T,2)=2$ by Theorem \ref{thm:binary.summary}. 
Additionally, using the complement, we get that the probability that $T_n$ does \emph{not} have property $\pi$ equals $P(T_n: \bar{\pi})=1- \frac{n}{2n-3} = \frac{n-3}{2n-3}$, and for trees $T$ that do not have property $\pi$, we have $m(T,2)=2n-4$ and $fm(T,2)=2$ by Theorem \ref{thm:binary.summary}.

In summary, this leads to:

\begin{align*}
  E_{unif}[m(T_n,2)]&=  \frac{n}{2n-3} \cdot (2n-2) +  \frac{n-3}{2n-3} \cdot (2n-4)\\
  &= 2n-3+\frac{3}{2n-3},
\end{align*}

and

\begin{align*}
  E_{unif}[fm(T_n,2)]&=  \frac{n}{2n-3} \cdot 2 +  \frac{n-3}{2n-3} \cdot 4\\
  &= 3 - \frac{3}{2n-3}.
\end{align*}

This completes the proof. 

\end{proof}

Next, we consider the Yule model.

\begin{proposition}\label{cor_binaryYule} 
Let $n \geq 2$ and let $T_n$ be a phylogenetic tree with $n$ leaves sampled under the Yule model. Then,
$$E_{Yule}[m(T_n,2)]=(2n-4)+2\cdot  \prod\limits_{i=3}^{n-1}\frac{i-1}{i},$$

and 

$$E_{Yule}[fm(T_n,2)]=4-2\cdot  \prod\limits_{i=3}^{n-1}\frac{i-1}{i}.$$
\end{proposition}

\begin{proof} We first calculate the probability $P(T_n:\pi)$ that a tree $T_n$ with $n$ leaves sampled under the Yule model has one leaf adjacent to the root, which we again refer to as having property $\pi$. We have $P(T_n:\pi)=\prod\limits_{i=3}^{n-1}\frac{i-1}{i}$, which can be seen as follows: For $n=2$ or $n=3$, there is only one tree shape, and this tree shape has property $\pi$, so all phylogenetic trees with $n=2$ or $n=3$ have said property. In these cases, the product is empty, which leads to the desired probability of 1. If $n>3$, however, the unique single leaf present in the tree that is adjacent to the root needs to be avoided in each step; all other leaves can be selected for speciation, i.e., for splitting into two lineages. This means that in each step, starting with $i=3$ leaves, there are $i-1$ leaves that can be chosen to split, whereas one specific leaf is forbidden. This leads to a probability of $\frac{i-1}{i}$ in each step, giving the above product as the total probability of generating a tree with $n$ leaves while fulfilling property $\pi$. Now, by Theorem \ref{thm:binary.summary}, all trees $T$ that fulfill property $\pi$ have $m(T,2)=2n-2$ and $fm(T,2)=2$. 

On the other hand, the probability that $T_n$ does \emph{not} have property $\pi$ is simply $1-P(T_n : \pi)=1-\prod\limits_{i=3}^{n-1}\frac{i-1}{i}$. In this case, again by Theorem \ref{thm:binary.summary}, the corresponding trees have $m(T,2)=2n-4$ and $fm(T,2)=4$.

In summary, this leads to the following expected values: 

\begin{align*}
  E_{Yule}[m(T_n,2)]&=  \prod\limits_{i=3}^{n-1}\frac{i-1}{i} \cdot (2n-2) +  \left(1-\prod\limits_{i=3}^{n-1}\frac{i-1}{i}\right) \cdot (2n-4)\\
  &=(2n-4)+((2n-2)-(2n-4))\cdot  \prod\limits_{i=3}^{n-1}\frac{i-1}{i}\\
  &=(2n-4)+2\cdot  \prod\limits_{i=3}^{n-1}\frac{i-1}{i},
\end{align*}

and

\begin{align*}
  E_{Yule}[fm(T_n,2)]&=  \prod\limits_{i=3}^{n-1}\frac{i-1}{i} \cdot 2 +  \left(1-\prod\limits_{i=3}^{n-1}\frac{i-1}{i}\right) \cdot 4\\
  &=4-2\cdot  \prod\limits_{i=3}^{n-1}\frac{i-1}{i},
\end{align*}

which completes the proof.

\end{proof}

We now turn our attention to the famous tree reconstruction criterion Maximum Parsimony, which turns out to be closely related to some of the monophyly concepts. 

\section{Connections of monophyletic characters to parsimony}\label{s:pars}

In this section, it is our aim to state relations between monophyly and the established parsimony concept. 

Note that characters only act on leaves of the tree, or, by means of the label set, on the present-day species assigned to them. In order to investigate monophyletic groups, we can also consider assignments of states to inner vertices. In this regard, we consider \emph{extensions} of characters. An extension $g_f$ of $f$ is a function $g_f:V(T)\rightarrow C$ with ${g_f}_{\vert X}=f$. In other words: $g_f$ assigns traits to \emph{all} vertices of $T$, not just to the leaves, but it agrees with $f$ on the leaves.

In this section, we will relate monophyletic characters to the well-known \emph{parsimony criterion}, which is often used in mathematical phylogenetics to reconstruct ancestral traits from character data. In this regard, for a given tree $T$ and character $f$, we call $ch(g_f) = \vert \{ \{u,v\} \in E, \, g_f(u) \neq g_f(v)\} \vert$ the \emph{changing number} or \emph{substitution number} of an extension $g_f$ on $T$.

The main concept needed for ancestral state reconstruction is \emph{maximum parsimony (MP)}, which is based on the so-called \emph{parsimony score} $l(f,T)$ of a character $f$ on a tree $T$. Here, $l(f,T) = \min\limits_{g_f} ch(g_f,T)$, where the minimum runs over all extensions $g_f$ of $f$ on $T$. 

Note that given a tree $T$ and a character $f$, the parsimony score can be efficiently calculated in linear time using the well-known \emph{Fitch algorithm} \cite{Fitch,Hartigan1973}.
This algorithm assigns a set of states to all inner vertices and minimizes the required number of changes. It is based on Fitch's parsimony operation which we explain now. Therefore, let $C$ be a non-empty finite set of character states or traits and let $A,B \subseteq C$. Then, Fitch's parsimony operation $*$ is defined by:
$$A*B \coloneqq \begin{cases}
A \cap B, & \text{if } A \cap B \neq \emptyset, \\
A \cup B, & \text{otherwise.}
\end{cases}$$

Using this operation, the Fitch algorithm works as follows. Given a tree $T$ and  character $f$, in the initialisation step, the algorithm assigns a set to each leaf containing precisely the state assigned to this leaf by $f$. For instance, if $f(1)=\alpha$, leaf $1$ gets assigned set $S(f,T,1)=\{\alpha\}$. The algorithm then proceeds as follows: In each step, it considers all vertices $v$ whose two children have already been assigned a set, say $A$ and $B$. Then, $v$ is assigned the set $S(f,T,v):=A*B$. This step is continued \enquote{upwards} (i.e., from the leaves to the root) along the tree until the root $\rho$ is assigned a set, which is denoted by $S(f,T,\rho)$. Note that throughout this section, when we want to describe the set $S$ assigned to a vertex $v$, by a slight abuse of notation, we often write $S(v)$ instead of $S(f,T,v)$ whenever there is no ambiguity. Ultimately, the parsimony score $l(f,T)$ then simply equals the number of times the Fitch operation had to use the union instead of the intersection \cite{Fitch,Hartigan1973}. Note that for a given tree $T$ with leaf set $X$ and any character $f$ on $X$ with $|f(X)|=r$, we have $l(f,T)\geq r-1$ (basically, this is due to the fact that the root can employ only one trait, and all other traits require at least one change). If we have $l(f,T)=r-1$, $f$ is said to be \emph{convex} or \emph{homoplasy-free} on $T$.

In particular, in this section, we are aiming at providing an algorithm which determines both the parsimony score and the \enquote{monophyly type} of a given character on a given tree (here, the monophyly type can be \emph{(non-relevant) monophyletic}, \emph{relevant monophyletic}, \emph{ (non-relevant) fully monophyletic}, \emph{relevant fully monophyletic} or \emph{not monophyletic}). 

However, before we state this algorithm, we will present some properties of fully monophyletic characters. We start with the following proposition.

\begin{proposition}\label{prop:convex} Let $f$ be a character on $X$ and let $T$ be a rooted binary phylogenetic $X$-tree. Then, if $f$ is fully monophyletic (relevant or not) with regards to $T$, it is also convex on $T$. 
\end{proposition}

\begin{proof}
We need to show that $l(f,T)=|f(X)|-1$. By definition of a fully monophyletic character, every leaf $\ell \in X$ is a member of some monophyletic group (in the relevant case: of size at least 2, otherwise it may be a singleton leaf). The monophyletic group of a state $a$ has the property that there is a vertex $v$ in $V(T)$ which gets assigned set $\{a\}$ by the Fitch algorithm and that all of its descendants also get assigned set $\{a\}$. Moreover, $a$ does not occur in any other set outside of subtree $T_v$ induced by $v$ in the Fitch algorithm. This, however, shows that no unions are taken within $T_v$, so such subtrees do not induce any changes. Therefore, we can replace each monophyletic group of a state $a$ by a single leaf labelled $a$. This leads to a smaller tree $T'$ with $|f(X)|$ many leaves, all in different states. This shows that obviously we need at least $|f(X)|-1$ many changes; and we can reach this number precisely by choosing one state $a\in f(X)$ for all inner vertices of $T'$. This induces $|f(X)|-1$ changes on the edges leading to the non-$a$-leaves of $T'$ or, equivalently, to the monophyletic groups in $T$. This completes the proof. 
\end{proof}

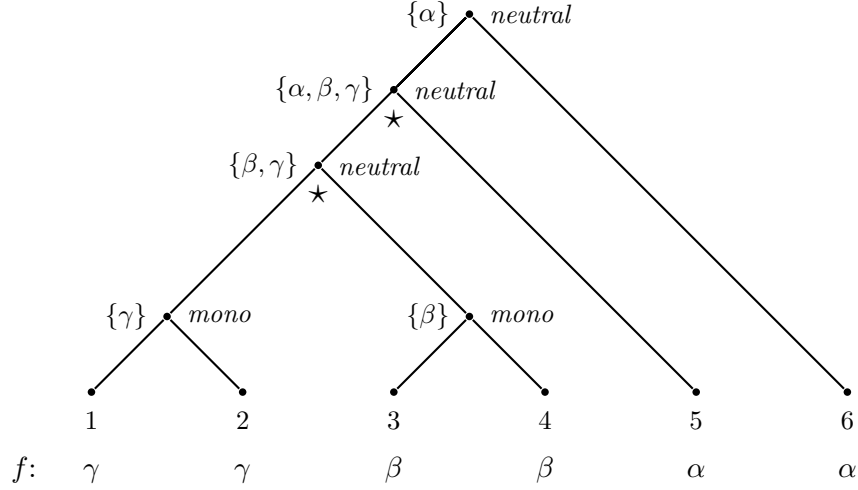
\begin{figure}[htbp]
\centering

\begin{tikzpicture}[
    sdot/.style={circle,fill,radius=1pt,inner sep=1pt},
    thick
]

\tikzset{
  char/.style={
    font=\large,
    text height=1.5ex,
    text depth=0.25ex
  }
}

\node[sdot] (1) at (0,0) {}; 
\node[sdot] (2) at (2,0) {}; 
\node[sdot] (3) at (4,0) {}; 
\node[sdot] (4) at (6,0) {}; 
\node[sdot] (5) at (8,0) {}; 
\node[sdot] (6) at (10,0) {}; 

\node[sdot] (A) at (1,1) {}; 
\node[sdot] (B) at (5,1) {}; 
\node[sdot] (C) at (3,3) {}; 
\node[sdot] (D) at (4,4) {}; 
\node[sdot] (E) at (5,5) {}; 

\draw (1)--(A)--(2);
\draw (3)--(B)--(4);
\draw (A)--(C)--(B);
\draw (E)--(6);
\draw (D)--(5);
\draw (C)--(E)--(D);

\node[below=2pt of 1] {1};
\node[below=2pt of 2] {2};
\node[below=2pt of 3] {3};
\node[below=2pt of 4] {4};
\node[below=2pt of 5] {5};
\node[below=2pt of 6] {6};

\node[left=2pt of A] {$\{\gamma\}$}; 
\node[right=2pt of A] {\emph{mono}}; 
\node[left=2pt of B] {$\{\beta\}$};
\node[right=2pt of B] {\emph{mono}};
\node[left=2pt of C] {$\{\beta,\gamma\}$};
\node[right=2pt of C] {\emph{neutral}};
\node[below=2pt of C] {\LARGE $\star$};
\node[left=2pt of D] {$\{\alpha,\beta,\gamma\}$};
\node[right=2pt of D] {\emph{neutral}};
\node[below=2pt of D] {\LARGE $\star$};
\node[left=2pt of E] {$\{\alpha\}$};
\node[right=2pt of E] {\emph{neutral}};

\node[below=20pt of 1, char] (f) {$\gamma$};
\node[below=20pt of 2, char] {$\gamma$};
\node[below=20pt of 3, char] {$\beta$};
\node[below=20pt of 4, char] {$\beta$};
\node[below=20pt of 5, char] {$\alpha$};
\node[below=20pt of 6, char] {$\alpha$};
\node[left=10pt of f]{\large $f$:};
\end{tikzpicture}
\caption{Example of a relevant monophyletic, but not fully monophyletic, character $f$ on a rooted binary phylogenetic tree $T$ that is convex on $T$. The sets of character states assigned to the inner vertices are the Fitch sets (see also Lines~11 and~19 of Algorithm~\ref{alg:mono}), while the labels \emph{mono} and \emph{neutral} indicate the corresponding vertex status (see the while loop of Algorithm~\ref{alg:mono}). The two union vertices of the Fitch algorithm are marked with asterisks, implying that $l(f,T)=2=3-1=|f(X)|-1$, which shows that $f$ is convex on $T$, even though $f^{-1}(c)$ does not induce a monophyletic group. Indeed, after the execution of the while loop of Algorithm~\ref{alg:mono}, we obtain $m_\alpha=2$ and $m_\beta=m_\gamma=1$ (with $n_\alpha=n_\beta=n_\gamma=2$). Consequently, Line~28 applies, and the character is classified as relevant monophyletic.}
\label{fig_convex}
\end{figure}

Next, we observe that the reverse of Proposition \ref{prop:convex} is not always true: There exist convex characters which are not fully monophyletic. An example is given in Figure \ref{fig_convex}, where the root position is such that state $\alpha$ does not form a monophyletic group. However, note that the depicted character is relevant monophyletic due to the monophyletic group induced by state $\beta$ (and $\gamma$, respectively). 

However, the tree depicted in Figure \ref{fig_convex} could be re-rooted to make it relevant fully monophyletic: for instance, the root could be placed on the edge leading to cherry $[1,2]$. But there are examples of unrooted trees together with a convex character $f$ for which no edge provides a proper root position to make the character fully monophyletic. An example is depicted in Figure \ref{fig_convex2}. 

\begin{figure}[htbp]
\centering
\begin{tikzpicture}[
    sdot/.style={circle,fill,radius=1pt,inner sep=1pt},
    thick
]

\tikzset{
  char/.style={
    font=\large,
    text height=1.5ex,
    text depth=0.25ex
  }
}

\node[sdot] (1) at (1,3) {}; 
\node[sdot] (2) at (1,1) {}; 
\node[sdot] (3) at (3,1) {}; 
\node[sdot] (4) at (4,1) {}; 
\node[sdot] (5) at (6,3) {}; 
\node[sdot] (6) at (6,1) {}; 

\node[sdot] (A) at (2,2) {}; 
\node[sdot] (B) at (3,2) {}; 
\node[sdot] (C) at (4,2) {}; 
\node[sdot] (D) at (5,2) {}; 

\draw (1)--(A)--(2);
\draw (A)--(B)--(3);
\draw (B)--(C)--(4);
\draw (C)--(D);
\draw (5)--(D)--(6);

\draw[line width=2pt] (A)--(B);
\draw[line width=2pt] (C)--(D);

\node[left=2pt of 1] {1};
\node[left=2pt of 2] {2};
\node[below=2pt of 3] {3};
\node[below=2pt of 4] {4};
\node[right=2pt of 5] {5};
\node[right=2pt of 6] {6};

\node[left=15pt of 1] {$\alpha$};
\node[left=15pt of 2] {$\alpha$};
\node[below=15pt of 3] {$\beta$};
\node[below=15pt of 4] {$\beta$};
\node[right=15pt of 5] {$\gamma$};
\node[right=15pt of 6] {$\gamma$};
\end{tikzpicture}
\caption{Example of an unrooted tree $T$ together with a convex character $f$, which has no edge on which the root could be placed to make $f$ fully monophyletic. Note that the highlighted edges are the ones showing that $l(f,T)=2$, which together with $|f(X)|=3$ shows that $f$ is indeed convex on $T$.}\label{fig_convex2}
\end{figure}
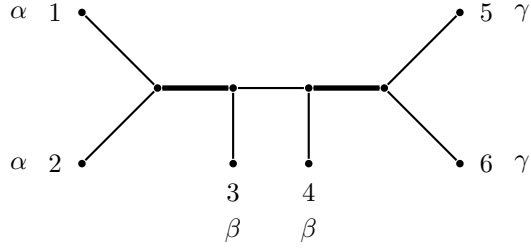

\par\vspace{0.5cm}

Next, we consider \emph{forbidden} vertices of $T$ with respect to a character $f$. Consider an inner vertex $u$ with children $v$ and $w$. Let $S(u)$, $S(v)$, and $S(w)$ denote the sets assigned to $u$, $v$, and $w$, respectively, by the Fitch algorithm. Then, we call $u$ \emph{forbidden} if $S(v)$ or $S(w)$ (or both) are unions and if $S(u)\subseteq S(v)$ or $S(u)\subseteq S(w)$. In Figure \ref{fig_convex}, the root is a forbidden vertex, because one of its children is a union vertex (highlighted with an asterisk), and the set assigned to the root is contained in said union.

\begin{proposition}\label{prop:forbidden} Let $f$ be a character on $X$ and let $T$ be a rooted binary phylogenetic $X$-tree. Then, if there are no forbidden vertices of $T$ with respect to $f$, $f$ is fully monophyletic. If additionally we have $|f^{-1}(a)|\geq 2$ for all $a \in f(X)$, then $f$ is relevant fully monophyletic.
\end{proposition}

\begin{proof} We start with the first statement and, seeking a contradiction, assume that it is not true. That is, we assume that there exists a rooted binary phylogenetic $X$-tree $T$ and a character $f$ on $X$ such that there is no forbidden vertex, but $f$ is not fully monophyletic. This must imply that there is a state $a\in f(X)$ that does not form a monophyletic group -- in particular, we must have $|f^{-1}(a)|\geq 2$. 

Let $w$ denote the lowest common ancestor of all $a$-vertices in $T$ and consider $T_w$. Then, by construction $T_w$ must have at least one $a$-leaf in each of its maximum pendant subtrees $T_w=(T_w^1,T_w^2)$.

Now consider an arbitrary run of the Fitch algorithm and fix the order in which the inner vertices of $T_w^i$ are considered (for each $i\in \{1,2\}$). Then, as $a$ does not form a monophyletic group in $T$, there must be a first inner vertex $u_i$ in $T_w^i$ which gets assigned $S(u_i)$ with $a\in S(u_i)$ and $S(u_i)$ being a union vertex (as an example, consider state $\alpha$
and the parent of leaf $5$ in Figure \ref{fig_convex}). Now, as $u$ is not a forbidden vertex by assumption, the parent of $u$ in $T_w^i$ has to be another union vertex (consisting of a strict superset of $S(u_i)$). For the same reasons, the same must be true for all ancestors of $u_i$ -- so the path from $u_i$ to the root of $T_w^i$ is such that the sets $S$ assigned by the Fitch algorithm are all unions and they get strictly larger in each step. 

Now consider $w$: Its children are the roots of $T_w^1$ and $T_w^2$, which are assigned union sets containing $a$ by construction. Thus, $w$ will be assigned an intersection containing $a$, thus $S(w)$ is a subset of both of the $S$-sets of its children. Thus, $w$ is a forbidden vertex, contradicting our assumption.

This shows that if $f$ is not fully monophyletic, the Fitch algorithm will find a forbidden vertex. The second statement then follows by the definition of relevant fully monophyletic characters. This completes the proof.
\end{proof}

We end this section with presenting Algorithm \ref{alg:mono} which can be regarded as an extension of the Fitch algorithm. In particular, it traverses all vertices and all states employed by the character (which is at most $n$) only once and thus works in time $\mathcal{O}(n)$, and it returns both the parsimony score and the monophyly type of the character on the tree (\emph{(non-relevant) monophyletic}, \emph{relevant monophyletic}, \emph{(non-relevant) fully monophyletic}, \emph{relevant fully monophyletic} or \emph{not monophyletic}). Recall that $\mathring{V}(T)$ refers to the set of internal vertices of $T$.

\begin{algorithm}
\caption{Detecting monophyly types}\label{alg:mono}
\KwInput{rooted binary phylogenetic $X$-tree $T$ and character $f$ on $X$ with $|X|=n$ and $f(X)=\{c_1,\ldots,c_m\}$ for some $m\leq n$.}
\KwOutput{$l(f,T)$ and monophyly type of $f$ on $T$}
$n_{c_i} \gets |\{x \in X: f(x)=c_i\}| \ \forall i=1,\ldots,m$\; 

$m_{c_i} \gets n_{c_i} \ \forall i=1,\ldots,m$\;

$status(x) \gets mono \ \forall x \in X$\;
$status(v) \gets unvisited \ \forall v \in \mathring{V}(T)$\;
$S(x)=\{f(x)\} \ \forall x \in X$\;
$S(v)=\emptyset \ \forall v \in \mathring{V}(T)$\;

$l(f,T) \gets 0$\;
$output \gets \emptyset$

\While{$\exists v\in V(T)$ with $status(v)=unvisited$}{
Let $v$ be such that $status(v)=unvisited$ and with children $w_1$ and $w_2$ with $status(w_1)\neq unvisited$ and $status(w_2)\neq unvisited$\;

  \eIf{$S(w_1)\cap S(w_2)\neq \emptyset$}{
    $S(v) \gets S(w_1)\cap S(w_2)$\;
    \eIf{$status(w_1)=status(w_2)=mono$}{
    $status(v)\gets mono$\;
    $m_{a}\gets m_a-1$ \ for $a \in S(v)$\; 
  }{$status(v)\gets neutral$\;}
  
  }{$S(v)\gets S(w_1)\cup S(w_2)$\;
    $l(f,T)\gets l(f,T)+1$\;
    $status(v) \gets neutral$;
  }
}
\eIf{$m_{c_i}=1 \ \forall i=1,\ldots,m$}{

\eIf{$n_{c_i}\geq 2 \ \forall i=1,\ldots,m$}{
$output \gets \mbox{relevant fully monophyletic}$\;}{$output \gets \mbox{fully monophyletic (non-relevant)}$\;}}
{\eIf{$\exists i\in \{1,\ldots,m\}$ with $m_{c_i}=1$ and $n_{c_i}\geq 2$}{
$output \gets \mbox{relevant  monophyletic}$\;}{\eIf{$\exists i\in \{1,\ldots,m\}$ with $m_{c_i}=1$}{$output \gets \mbox{monophyletic (non-relevant)}$\;}{$output \gets \mbox{not monophyletic}$\;}}
}
\Return{$l(f,T)$ \mbox{and} $output$.}
\end{algorithm}

We now argue why Algorithm \ref{alg:mono} is correct. First, lines 9--23 describe the traversal through all inner vertices which determines the parsimony score based on unions and intersections of the $S(v)$ sets just like the well-known Fitch algorithm. However, additionally this traversal also calculates $m_a$ for each character state $a \in f(X)$, where $m_a$ ultimately will determine if $a$ forms a monophyletic group in $T$. This is done as follows: In the beginning, each leaf gets assigned the status \emph{mono}, whereas all inner vertices have the status \emph{unvisited}. Once a vertex has been visited, the \emph{unvisited} status will change to either \emph{mono} or \emph{neutral}, it never goes back to \emph{unvisited}. At the same time, the counter $m_a$, which at the initialisation step in line 2 was assigned the number of times state $a$ occurs in $f(X)$, gets decreased precisely if two monophyletic vertices have a non-empty $S$-intersection. It is crucial to note that this can only happen if the two $S$-sets have size 1 and are identical. In this case, $m_a$ gets decreased by one and the currently considered vertex gets the status \emph{mono}. Otherwise, $m_a$ remains unchanged and the vertex gets status \emph{neutral} to indicate that it is not monophyletic but has been visited (the latter is important for the while-loop's stopping criterion). 

So when the while loop is finished, crucially $m_a$ contains the information whether $a$ is a monophyletic state (if $m_a=1$), because $m_a$ describes the number of subtrees only employing state $a$ (i.e., potential monophyletic groups in state $a$). Thus, lines 22--34 merely need to sort the information for all $a\in f(X)$ accordingly in order to determine the monophyly state of the character in the variable called \emph{output}. In the end, both $l(f,T)$ and \emph{output} are returned.

Note that the Fitch algorithm does not require the traversal through all states of $f(X)$ described in lines 22--34. In fact, $l(f,T)$ could be returned directly after the while-loop, i.e., in line 22. Thus, our algorithm requires $|f(X)|$  more steps than the well-known Fitch algorithm -- but as $|f(X)|\leq n$, the entire algorithm's runtime is still in $\mathcal{O}(n)$.

\smallskip
An example is shown in Figure~\ref{fig_convex}. In this example, $n_\alpha=n_\beta=n_\gamma=2$, since each of the states $\alpha$, $\beta$, and $\gamma$ occurs twice. Running the algorithm assigns a Fitch set and a status to every inner vertex of the tree $T$, as illustrated in the figure (for simplicity, Fitch sets and statuses are not shown for the leaves). For the vertices marked with asterisks, Lines~19--20 of the algorithm are executed, yielding $l(f,T)=2$. In addition, for states $\beta$ and $\gamma$, Line~15 of Algorithm~\ref{alg:mono} is executed once each, so that $m_\beta=m_\gamma=1$, while $m_\alpha=2$. Consequently, Line~28 of the algorithm applies, and the character $f$ is classified as relevant monophyletic.

\section{Discussion and outlook}

Monophyly can be treated as a combinatorial property of characters, with multiple distinct notions of what it means to be monophyletic. These notions admit exact enumeration and expectation calculations and connect naturally to convexity and parsimony. 

In particular, we have introduced four types of monophyletic characters, and provided general formulae for counting the number of monophyletic characters for two of them. We have directly calculated the monophyly number for several special cases, and provided the expected number of monophyletic binary characters for trees with $n$ leaves for each of them. Finally, we linked the concept of monophyly with the well-known theories of convex characters and parsimony scores, and provided an algorithm for calculating parsimony score and monophyly type in $\mathcal{O}(n)$ time.

The existence of several inequivalent notions of monophyly, with different properties, highlights that monophyly according to our notions is not a single combinatorial phenomenon. Rather, different notions capture different structural relationships between character states and tree topology, and may therefore be appropriate in different biological or statistical settings.

Our results show that the monophyly number depends strongly on tree shape, and in particular indicate that monophyly number may provide a novel quantitative descriptor of tree shape. 

Surprisingly, our results suggest that classical notions of tree balance may not fully explain variation in monophyly number. For example, in the binary case, the presence or absence of a leaf adjacent to the root completely determines the extremal values of the monophyly number, regardless of the remaining internal structure. This indicates that local structural features may have a stronger influence than global balance measures. Nevertheless, it remains natural to ask whether familiar tree families such as caterpillars and fully balanced trees exhibit extremal behaviour under our notions of monophyly.

One of the main remaining theoretical questions of course, is whether it is possible to find a characterisation of full monophyly numbers in the spirit of Theorem \ref{thm:gen.theorem}. This seems deeply related to determining the existence of a set of (up to) $c$ internal vertices $\{v_1,\dots,v_c\}$ so that $L(v_1),\dots,L(v_c)$ partitions the set of leaves.

Finally, we believe that -- based on their expected values under different models, as discussed in Section \ref{s:expmon} -- each of the monophyly numbers have potential to form the basis for statistical testing, in which one can observe the prevalence of monophyly for (sets of) trees compared to the base expected level of monophyly. An interesting avenue for future research would thus be to investigate the expected values of these numbers under more complex models. 

\section*{Acknowledgements} 
The authors wish to thank Falk Nagies for providing the initial inspiration for the project. The authors additionally wish to thank Andrew Francis for helpful discussions and insights during the project. 
Parts of this material are based upon work supported by the National Science Foundation under Grant No. DMS-1929284 while MF and KW were in residence at the Institute for Computational and Experimental Research in Mathematics in Providence, RI, during the Theory, Methods, and Applications of Quantitative Phylogenomics semester program. MH's research was supported by the Australian Government through the Australian Research Council's Discovery Projects funding scheme (project DP260102678). The views expressed herein are those of the authors and are not necessarily those of the Australian Government or Australian Research Council.

\section*{Conflict of interest} The authors herewith certify that they have no affiliations with or involvement in any
organisation or entity with any financial (such as honoraria; educational grants; participation in speakers’ bureaus;
membership, employment, consultancies, stock ownership, or other equity interest; and expert testimony or patent-licensing
arrangements) or non-financial (such as personal or professional relationships, affiliations, knowledge or beliefs) interest in the subject matter discussed in this manuscript.

\section*{Data availability statement} 
Data sharing is not applicable to this article as no new data were created or analysed in this study.

\bibliographystyle{plain}
\bibliography{monobib}   

\end{document}